\documentclass[sigconf]{acmart}

\usepackage{soul}
\usepackage{url}
\usepackage{hyperref}
\usepackage[utf8]{inputenc}
\usepackage{caption}
\usepackage{xcolor}
\usepackage{graphicx}
\usepackage{caption}
\usepackage{courier}
\usepackage{multirow}
\usepackage{subcaption}
\usepackage{color}
\usepackage{todonotes}
\usepackage{amssymb}
\usepackage{gensymb}
\usepackage{enumitem}
\usepackage[linesnumbered, ruled]{algorithm2e}
\usepackage{amsmath}

\copyrightyear{2020}
\acmYear{2020}
\setcopyright{acmlicensed}\acmConference[KDD '20]{Proceedings of the 26th ACM SIGKDD Conference on Knowledge Discovery and Data Mining}{August 23--27, 2020}{Virtual Event, CA, USA}
\acmBooktitle{Proceedings of the 26th ACM SIGKDD Conference on Knowledge Discovery and Data Mining (KDD '20), August 23--27, 2020, Virtual Event, CA, USA}
\acmPrice{15.00}
\acmDOI{10.1145/3394486.3403276}
\acmISBN{978-1-4503-7998-4/20/08}
\newtheorem*{problem}{Problem}

\newcommand{\feature}[1]{\texttt{#1}}
\newcommand{\snap}{{Snapchat}\xspace}

\newcommand{\ours}{{\textsf{FATE}}\xspace}
\newcommand{\tgcn}{{tGCN}\xspace}
\newcommand{\tlstm}{{tLSTM}\xspace}

\newcommand{\regionone}{\textsf{Region 1}\xspace}
\newcommand{\regiontwo}{\textsf{Region 2}\xspace}
\newcommand{\taskone}{\textsf{Task 1}\xspace}
\newcommand{\tasktwo}{\textsf{Task 2}\xspace}

\def \E {\mathcal{E}}
\def \G {\mathcal{G}}
\def \V {\mathcal{V}}
\def \N {\mathcal{N}}

\def \L {\mathcal{L}}

\def \R {\mathbb{R}}

\def \bg {\mathbf{G}}
\def \bt {\mathbf{T}}
\def \A {\mathbf{A}}

\def \F {\mathbf{F}}
\def \X {\mathbf{X}}

\def \h {\mathbf{h}}
\def \W {\mathbf{W}}
\def \S {\mathcal{S}}

\def \h {\mathbf{h}}
\def \x {\mathbf{x}}
\def \tx {\mathbf{\tilde{x}}}
\def \a {\mathbf{a}}
\def \b {\mathbf{b}}
\def \q {\mathbf{q}}
\def \f {\mathbf{f}}
\def \be {\mathbf{E}}
\def \e {\mathbf{e}}
\def \g {\mathbf{g}}
\title{Knowing your FATE: \underline{F}riendship, \underline{A}ction and \underline{T}emporal \underline{E}xplanations for User Engagement Prediction on Social Apps}

\author{Xianfeng Tang$^\dagger$, Yozen Liu$^{\ddagger}$, Neil Shah$^\ddagger$, Xiaolin Shi$^\ddagger$, Prasenjit Mitra$^\dagger$, Suhang Wang$^\dagger$}
\affiliation{
\institution{The Pennsylvania State University$^\dagger$, Snap Inc.$^\ddagger$}
}
\affiliation{
  \institution{\{xut10, pum10, szw494\}@psu.edu  \{yliu2, nshah, xiaolin\}@snap.com}
}
\authornote{Work done when Xianfeng Tang is an intern at Snap.}

\begin{document}

\begin{abstract}
  With the rapid growth and prevalence of social network applications (Apps) in recent years, understanding user engagement has become increasingly important, to provide useful insights for future App design and development.
  While several promising neural modeling approaches were recently pioneered for accurate user engagement prediction, their black-box designs are unfortunately limited in model explainability.
  In this paper, we study a novel problem of explainable user engagement prediction for social network Apps.
  First, we propose a flexible definition of user engagement for various business scenarios, based on future metric expectations.
  Next, we design an end-to-end neural framework, \ours, which incorporates three key factors that we identify to influence user engagement, namely \underline{f}riendships, user \underline{a}ctions, and \underline{t}emporal dynamics to achieve \underline{e}xplainable engagement predictions.  \ours is based on a tensor-based graph neural network (GNN), LSTM and a mixture attention mechanism, which allows for (a) predictive explanations based on learned weights across different feature categories, (b) reduced network complexity, and (c) improved performance in both prediction accuracy and training/inference time.
  We conduct extensive experiments on two large-scale datasets from \snap, where \ours outperforms state-of-the-art approaches by ${\approx}10\%$ error and ${\approx}20\%$ runtime reduction. We also evaluate explanations from \ours, showing strong quantitative and qualitative performance. 
\end{abstract}

\maketitle

\section{Introduction}
With rapid recent developments in web and mobile infrastructure, social networks and  applications (Apps) such as \snap and Facebook have risen to prominence.
The first priority of development of most social Apps is to attract and maintain a large userbase. Understanding user engagement plays an important role for retaining and activating users.
Prior studies try to understand the return of existing users 
using different metrics, such as churn rate prediction \cite{yang2018know} and lifespan analysis \cite{yang2010activity}.
Others model user engagement with macroscopic features (e.g., demographic information) \cite{althoff2015donor} and historical statistic features (e.g., user activities) \cite{lin2018ll}.
Recently, \citeauthor{liu2019characterizing} \cite{liu2019characterizing} propose using dynamic action graphs, where nodes are in-App actions, and edges are transitions between actions, to predict future activity using a neural model.

Despite some success, existing methods generally suffer from the following:
\textbf{(1)} They fail to model friendship dependencies or ignore user-user interactions when modeling user engagement. 
As users are connected in social Apps, their engagement affects each other \cite{subramani2003knowledge}. For example, active users may keep posting new contents, which attract his/her friends and elevate their engagement. Thus, it is essential to capture friendship dependencies and user interactions when modeling user engagement.
\textbf{(2)} Engagement objectives may differ across Apps and even across features. For example, an advertising team may target prediction of click-through-rate, while a growth-focused team may care about usage trends in different in-App functions. Therefore, the definition of user engagement must be flexible to satisfy different scenarios.
\textbf{(3)} Existing methods focus on the predicting user engagement accurately, but fail to answer \emph{why} a user engages (or not). Explaining user engagement is especially desirable, since it provides valuable insights to practitioners on user priorities and informs mechanism and intervention design for managing different factors motivating different users' engagement. However, to our knowledge, there are no explainable models for understanding user engagement.

To tackle the aforementioned limitations, we aim to use three key factors: friendship, in-App user actions, and temporal dynamics, to derive explanations for user engagement.
Firstly, since users do not engage in a vacuum, but rather with each other, we consider friendships to be key in engagement.  
For example, many users may be drawn to use an App because of their family and friends' continued use.  
Secondly, user actions dictate how a user uses different in-App features, and hints at their reasons for using the App. 
Thirdly, user behavior changes over time, and often obey temporal periodicity \cite{papapetrou2014social}. Incorporating periodicity and recency effects can improve predictive performance. 

In this work, we first propose measurement of user engagement based on the expectation of metric(s) of interests in the future, which flexibly handles different business scenarios.  Next, we formulate a prediction task to forecast engagement score, based on heterogeneous features identified from friendship structure, user actions, and temporal dynamics.  Finally, to accurately predict future engagement while also obtaining meaningful explanations, we propose an end-to-end neural model called \ours (\underline{F}riendship, \underline{A}ction and \underline{T}emporal \underline{E}xplanations).  In particular, our model is powered by (a) a friendship module which uses a tensor-based graph convolutional network to capture the influence of network structure and user interactions, and (b) a tensor-based LSTM \cite{guo2019exploring} to model temporal dynamics while also capturing exclusive information from different user actions.  \ours's tensor-based design not only improves explainablity aspects by deriving both local (user-level) and global (App-level) importance vectors for each of the three factors using attention and Expectation-Maximization, but is also more efficient compared to classical versions.  We show that \ours significantly outperforms existing methods in both accuracy and runtime on two large-scale real-world datasets collected from \snap, while also deriving high-quality explanations.  To summarize, our contributions are:
\begin{itemize}[leftmargin=*]
    \item We study the novel problem of explainable user engagement prediction for social network applications;
    \item We design a flexible definition for user engagement satisfying different business scenarios;
    \item We propose an end-to-end self-explainable neural framework, \ours, to jointly predict user engagement scores and derive explanations for friendships, user actions, and temporal dynamics from both local and global perspectives; and
    \item We evaluate \ours on two real-world datasets from \snap, showing ${\approx}10\%$ error reduction and ${\approx}20\%$ runtime improvement against state-of-the-art approaches.
\end{itemize}

\section{Related Work}
\subsection{User Behaviour Modeling}
 
Various prior studies model user behaviours for social network Apps.
Typical objectives include churn rate prediction,  return rate analysis, intent prediction, etc \cite{au2003novel,kawale2009churn,kapoor2014hazard,benson2016modeling,lo2016understanding,yang2018know,kumar2018did,liu2019characterizing} and anomaly detection \cite{shah2017flock,lamba2019modeling,shah2017many}.
Conventional approaches rely on feature-based models to predict user behaviours. They usually apply learning methods on handcrafted features.
For example, \citeauthor{kapoor2014hazard}\cite{kapoor2014hazard} introduces a hazard based prediction model to predict user return time from the perspective of survival analysis;
\citeauthor{lo2016understanding}\cite{lo2016understanding} extract long-term and short-term signals from user activities to predict purchase intent;
\citeauthor{trouleau2016just}\cite{trouleau2016just} introduce a statistical mixture model for viewer consumption behavior prediction based on video playback data.
Recently, neural models have shown promising results in many areas such as computer vision and natural language processing, and have been successfully applied for user modeling tasks \cite{elkahky2015multi,yang2018know,liu2019characterizing}. \citeauthor{yang2018know}\cite{yang2018know} utilize LSTMs \cite{hochreiter1997long} to predict churn rate based on historical user activities.
\citeauthor{liu2019characterizing}\cite{liu2019characterizing} introduce a GNN-LSTM model to analyze user engagement, where GNNs are applied on user action graphs, and an LSTM is used to capture temporal dynamics.
\textit{Although these neural methods show superior performance, their black-box designs hinder interpretability, making them unable to summarize the reasons
for their predictions, even when their inputs are meaningful user activities features.}


\subsection{Explainable Machine Learning}
Explainable machine learning has gain increasing attention in recent years \cite{gilpin2018explaining}.
We overview recent research on explainable GNN/RNN models, as they relate to our model design.
We group existing solutions into two categories.  The first category focuses on post-hoc interpretation for trained deep neural networks. One kind of model-agnostic approach learns approximations around the predictions, such as linear proxy model \cite{ribeiro2016should} and decision trees \cite{schmitz1999ann,zilke2016deepred}.
Recently, \citeauthor{ying2019gnn}\cite{ying2019gnn} introduce a post-hoc explainable graph neural network to analyze correlations between graph topology, node attributes and predicted labels by optimizing a compact
subgraph structure indicating important nodes and edges.
\textit{However, post-analyzing interpretations are computationally
inefficient, making it difficult to deploy on large systems. Besides, these methods do not help predictive performance.}
The second group leverages attention methods to generate explanations on-the-fly, and gained tremendous popularity due to their efficiency \cite{xu2018raim,guo2019exploring,choi2018fine,pope2019explainability,shu2019defend}. For example,  \citeauthor{pope2019explainability}\cite{pope2019explainability} extend explainability methods for convolutional neural networks (CNNs) to cover GNNs; \citeauthor{guo2019exploring}\cite{guo2019exploring} propose an interpretable LSTM architecture that distinguishes the contribution of different input variables to the prediction.
\textit{Despite these attention methods successfully provides useful explanations, they are typically designed for one specific deep learning architecture (e.g., LSTMs or CNNs). How to provide attentive explanations for hierarchical deep learning frameworks with heterogeneous input is yet under-explored.}

\section{Preliminaries} \label{prelim}
\begin{figure}[t]
    \centering
    \includegraphics[width=0.7\columnwidth]{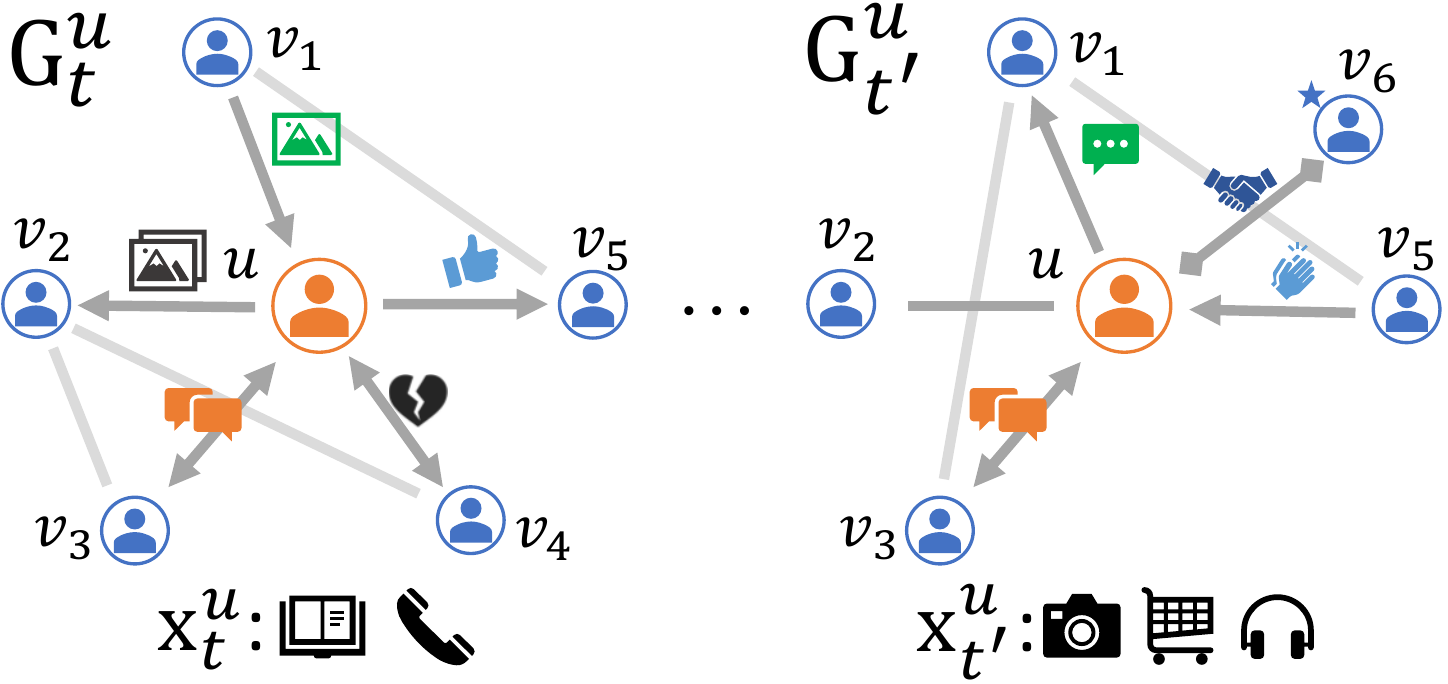}
    \vspace{-1em}
    \caption{User graphs are temporal, and capture friendship structure, user actions (node features), and user-user interactions (edge features) over various in-App functions.}
    \label{fig:problem_def}
    \vspace{-1.8em}
\end{figure}

First, we define notations for a general social network App.
We begin with the \textit{user} as the base unit of an App. Each user represents a registered individual. We use $u$ to denote a user.
We split the whole time period (e.g., two weeks) into equal-length continuous \textit{time intervals}. The length of time intervals can vary from hours to days.
The past $T$ time intervals in chronological order are denoted as $1, 2,\cdots, T$.
Users are connected by \textit{friendship}, which is an undirected relationship. Namely, if $u$ is a friend of $v$, $v$ is also a friend of $u$. Note that friendship is time aware, users can add new friends or remove existing friends at any given time.
Users can also use multiple in-App features, like posting a video, chatting with a friend, or liking a post on Facebook; we call these various \textit{user actions}.  We use a time-aware feature vector to represent the user action for each specific user.
A typical feature of social network Apps is in-App communication. By sending and receiving messages, photos, and videos, users share information and influence each other.  We call these \textit{user interactions}.

\textbf{User graph}: To jointly model user activities and social network structures, we define a temporal \textit{user graph} for every user at time $t$ as $\bg_t^u = (\V_t^u, \E_t^u, \X_t^u, \be_t^u)$.
Here $\V_t^u = \{u\} \cup \N_t(u)$ denotes the nodes in $\bg_t^u$,
where $\N_t(u)$ is a group of users related to $u$, 
the set of edges $\E_t^u$ represents friendships, nodal features $\X_t^u$ characterize user actions, 
and features on edges $\be_t^u$ describe user interactions.
Note that we split nodal features into $K$ categories, so that each category of features is aligned with a specific user action, respectively.  Thus, both the topological structure and the features of user graphs are temporal.
In particular, for any given node $u$, its feature vector (i.e., a row of $\X_t$) is represented by $\x_t^u = [\x_{t,1}^u, \cdots, \x_{t,K}^u]$, where $\x_{t,k}^u \in \R^{d_k}$ is the $k$-th category of features, and $[\cdot]$ denotes concatenation alongside the row.
There are many ways to define the graph structure.  One example of selecting $\bg$ is based on ego-networks, as shown in Figure \ref{fig:problem_def}; here, $\N_t(u)$ is the set of friends of $u$, which reduces the size of graph sharply compared to using the whole social network.  Each individual can take different actions in every time interval to control and use in-App functions.

\textbf{Defining user engagement}:
Because of the dynamism of user activities, social network structure, and the development of the App itself, the user engagement definition should be specified for every user and every time interval.  Besides, the primary focus of user engagement varies widely depending on the specific business scenario.
For example, Facebook may utilize login frequency to measure engagement, while \snap may use the number of messages sent.
Thus, user engagement requires a flexible definition which can meet different needs.
To tackle above challenges, we define user engagement score using the \textit{expectation of a metric of interest in the future}, as:
$
    e_t^u = \mathop{\mathbb{E}}(\mathcal{M}(u, \tau)| \tau \in [t, t + \Delta t]))
$, 
where $\mathcal{M}$ is the metric of interest, and $\Delta t$ denotes a future time period.
Both the metric and the time interval can be adjusted by scenario.

\begin{figure}[t]
    \vspace{-1em}
    \centering
    \includegraphics[width=.95\columnwidth]{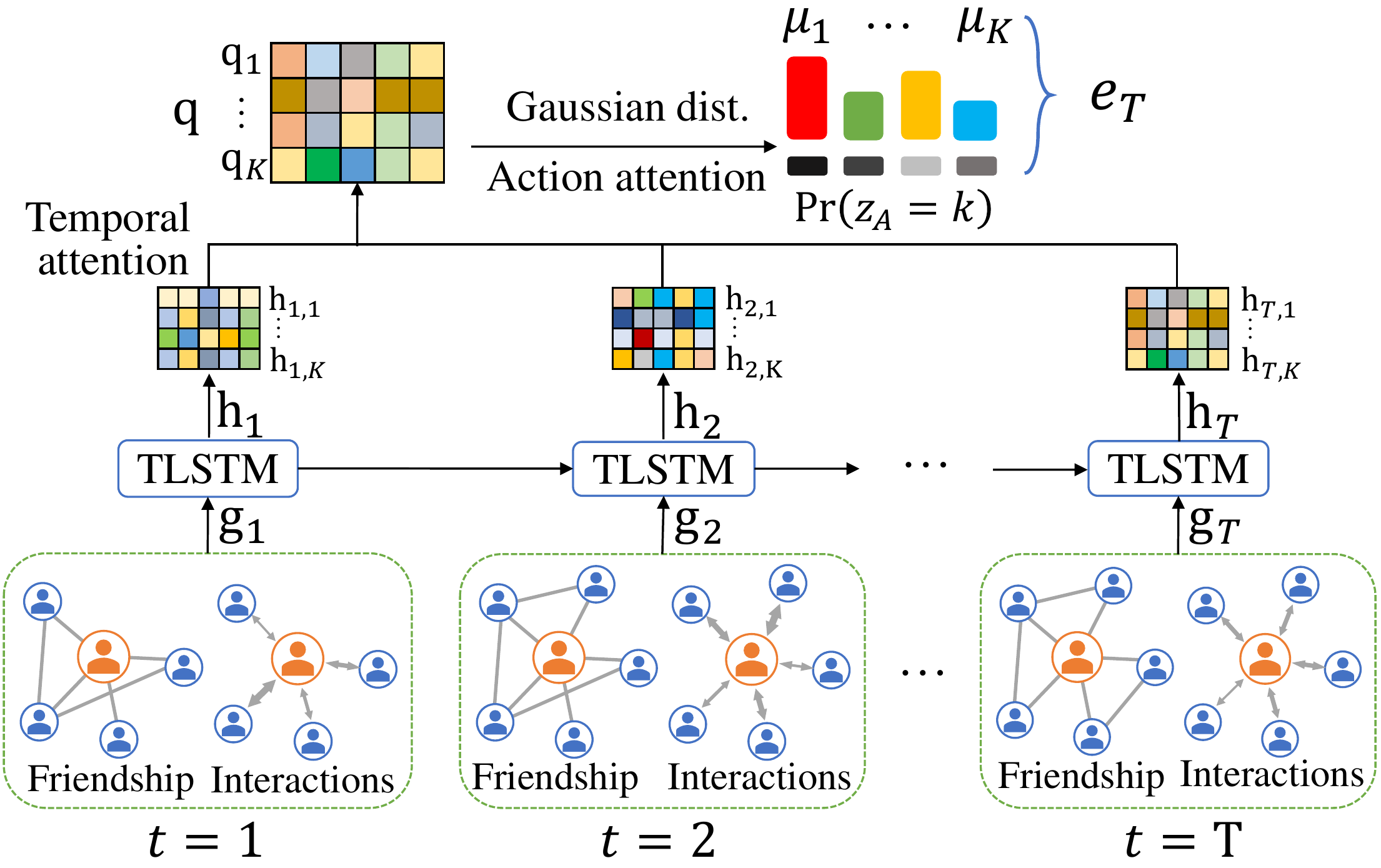}
    \vspace{-1em}
    \caption{Overall framework of \ours: \tgcn-based friendship modules capture local network structure and user interactions at each timestep, and \tlstm captures temporal dynamics for distinct user actions.  Finally, an attention mixture mechanism governs user engagement prediction.}
    \vspace{-1.5em}
    \label{fig:framework}
\end{figure}

\textbf{Explaining user engagement}: We identify three key factors that highly impact the user engagement, including user action, temporal dynamics, and friendship.
The interpretation is to derive importance/influence of these three factors for user engagement.
In particular, we aim at interpreting user engagement from both \textit{local} (i.e., for individual users) and \textit{global} (i.e., for the whole group of people, or even the entire App) perspectives.
The local interpretations for individual users are formulated as following vectors:
(1) User action importance $\A^u \in \R_{\geq 0}^K$, $\sum_{k=1}^{K} \A_k^u = 1$, which assigns each user action a score that reflects its contribution to user engagement.
(2) Temporal importance $\bt^u \in \R_{\geq 0}^{T \times K}$, $\sum_{t=1}^{T} \bt_{tk}^u = 1$ for $k=1,\cdots,K$, which identifies the importance of user actions over every time interval for the engagement;
(3) Friendship importance $\F^u \in \R_{\geq 0}^{|t \times \N_t(u)|}$, $\sum_{v \in \N_t(u)} \F_{tv}^u = 1$ for $t=1,\cdots,T$, which characterizes the contributions of friends to user engagement of $u$ over time.
For user action and temporal dynamics, we also derive explanations from a global view since they are shared by all users.
Specifically, we formulate (1) global user action importance $\A^* \in \R_{\geq 0}^K$, $\sum_{k=1}^{K} \A^*_k = 1$ and (2) global temporal importance $\bt^* \in \R_{\geq 0}^{T \times K}$, $\sum_{t=1}^{T} \bt^*_{tk} = 1$ for $k=1,\cdots,K$. Compared to local explanations which help understand individual user behaviors, global explanations inform overall App-level user behaviors.

We pose the following problem formalization:
\vspace{-0.5em}
\begin{problem}[Explainable Engagement Prediction] \label{problem_def}
Build a framework that (a) for every user $u$, predicts the engagement score $e_T^u$ with explanations $\A^u$, $\bt^u$ and $\,\F^u$ based on the historical user graphs $\bg_1^u,\cdots,\bg_T^u$, and (b) generates global explanations $\A^*$ and $\bt^*$.
\end{problem}

\section{Our Approach: {\ours}}

We next introduce our proposed approach for explainable engagement prediction, \ours.  Firstly, \ours leverages specific designed friendship modules (bottom of Figure \ref{fig:framework}) to model the non-linear social network correlations and user interactions from user graphs of a given user as input. The friendship modules aggregate user graphs and generate representations for user graphs accordingly. These graph representations preserve exclusive information for every time interval and every user action.
Next, a temporal module based on tensor-based LSTM \cite{guo2019exploring} (\tlstm, middle part of Figure \ref{fig:framework}) is utilized to capture temporal correlations from graph representations.
Finally, a mixture of attention mechanisms (top of Figure \ref{fig:framework}) is deployed to govern the prediction of user engagement based on the output of \tlstm, while also jointly deriving importance vectors as explanations.  An illustration of the framework is given in Figure \ref{fig:framework}.
We discuss \ours in detail in the following text.

\subsection{Friendship Module}
\begin{figure}[t]
\vspace{-1em}
    \centering
    \includegraphics[width=0.9\columnwidth]{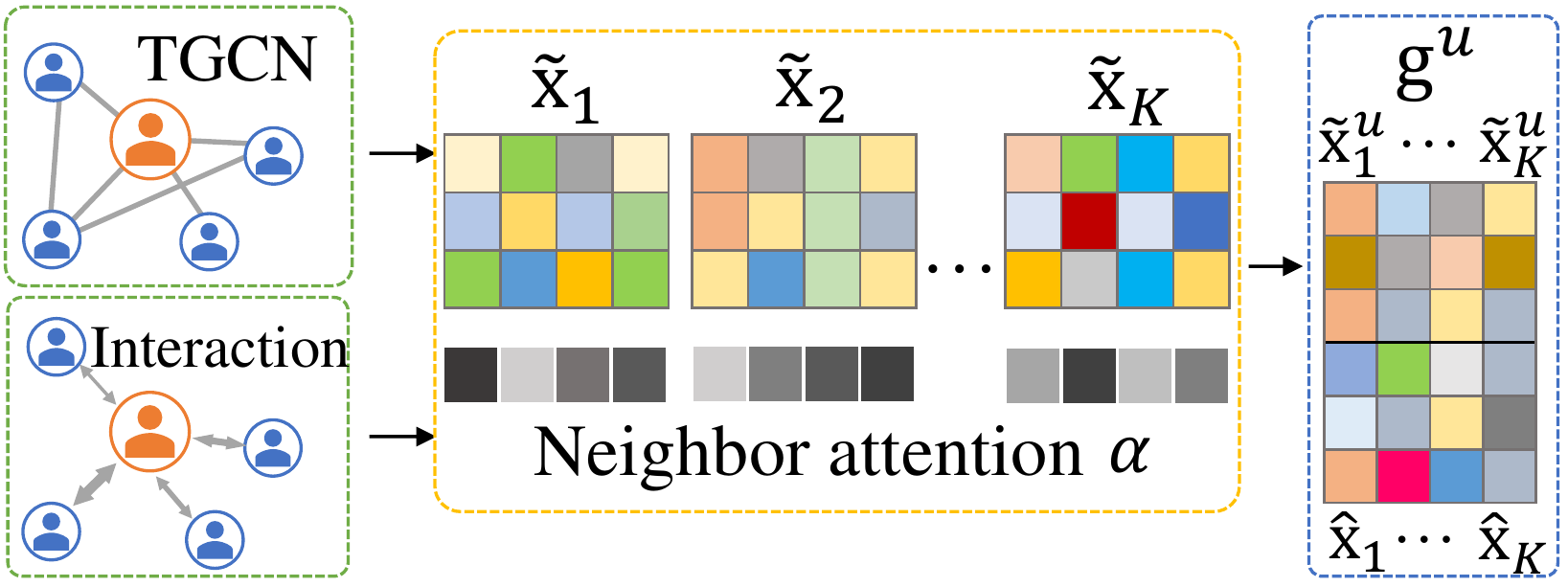}
    \vspace{-1.5em}
    \caption{Our proposed friendship module uses a tensor-based GCN with neighbor attention to generate user graph embeddings jointly from ego-networks and interactions.}
    \label{fig:tgcn}
    \vspace{-1.5em}
\end{figure}

As shown in Figure \ref{fig:tgcn}, the goal of the friendship module is to model the non-linear correlation of social network structure and user interactions in every user graph $\bg^u_t$.
Naturally, graph neural networks (GNNs) \cite{ma2019graph,ma2019multi,tang2020transferring,jin2020graph} can be applied to capture the dependencies of users. We choose the popular graph convolutional networks (GCNs) \cite{kipf2016semi} as our base GNN model. A GCN takes a graph as input, and encodes each node into an embedding vector. The embedding for each node is updated using its neighbor information on each layer of a GCN as:
\begin{equation} \label{eqn:gcn}
\small
    {\tx}^u = \sigma \left( \sum_{v \in \N(v)}\x^v\W \right),
\end{equation}
where $\x$ and $\tilde{\x}$  denote input feature and output embedding of the layer, respectively, $\W$ is a feature transformation matrix, and  $\sigma(\cdot)$ denotes a non-linear activation.

However, adopting vanilla GCN in our case is not ideal, because matrix multiplication in GCN mixes all features together. It is difficult to distinguish the importance of input features by looking at the output of a GCN layer.
To tackle this limitation, we propose a \textit{tensor-based GCN} (\tgcn), which uses a tensor of learnable parameters.
The updating rule of one \tgcn layer is:
\begin{equation} \label{eqn:tgcn}
\small
    \tilde{\x}^u = \sigma \left( \sum_{v \in \N(v)}  \x^v \otimes \mathcal{W} \right),
\end{equation}
where $\mathcal{W} = \{\W_1, \cdots, \W_K\}$, $\W_k \in \R^{d_k \times d^\prime}$, is a set of $K$ parameter matrices corresponding to each group of features, and  $\x^v \otimes \mathcal{W} =  [\x^v_1\W_1, \cdots, \x^v_K\W_K] \in \R^{K \times d^\prime}$, $\x^v_k\W_k \in \R^{1 \times d^\prime}$ maps each category of features from the input to the output space separately (as illustrated by different matrices in the middle part of Figure \ref{fig:tgcn}).
Note that each element (e.g. row) of the hidden matrix in a \tgcn layer encapsulates information exclusively from a certain category of the input, so that the following mixture attention can distinguish the importance of different user actions and mix exclusive information to improve prediction accuracy.
A \tgcn layer can be treated as multiple parallel vanilla GCN layers, where each layer is corresponding to one category of features that characterizes one user action.
Given a user graph input, We adopt a two-layer \tgcn to encode the friendship dependencies into node embedding:
\begin{equation}
\small
    \tilde{\X} = \sigma \left( \hat{\A}\sigma \left( \hat{\A} \X \otimes \mathcal{W}_0 \right) \otimes \mathcal{W}_1 \right),
\end{equation}
where $\hat{\A}$ is the symmetric normalized adjacency matrix derived from the input user graph, $\X$ are nodal features, and $\mathcal{W}_*$ are parameters.
As input features describe user actions, their exclusive information is preserved in the output of \tgcn as $\tilde{\X} = [\tilde{\X}_1, \cdots, \tilde{\X}_K] \in \R^{K \times d^\prime \times (|\N(v)| + 1)}$, which will be used later for generating engagement predictions and explanations.

\begin{figure}[t]
    \centering
        \begin{subfigure}[b]{.2\textwidth}
            \centering
            \includegraphics[width=\columnwidth]{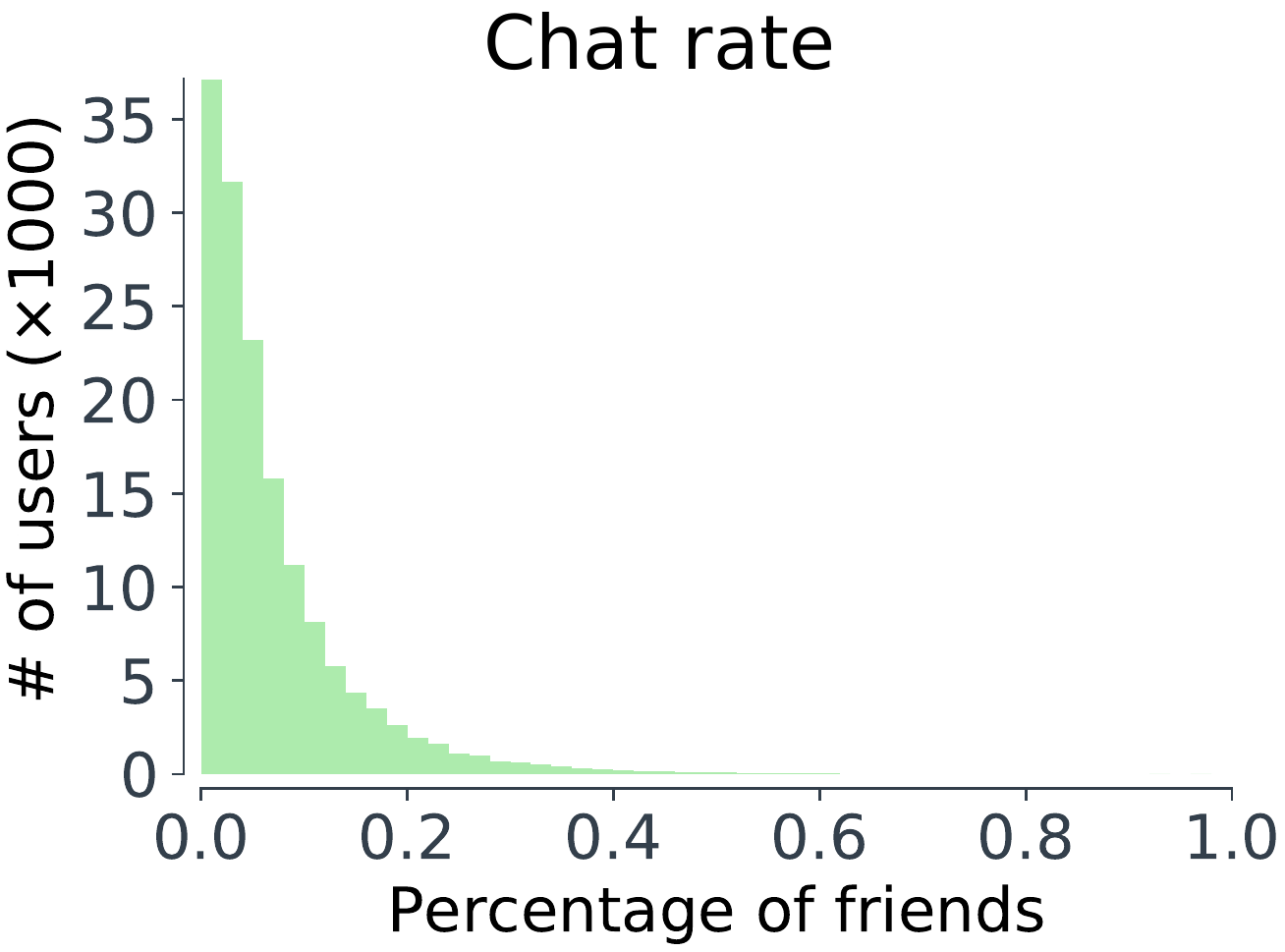}
        \end{subfigure}
        \begin{subfigure}[b]{.2\textwidth}  
            \centering 
            \includegraphics[width=\columnwidth]{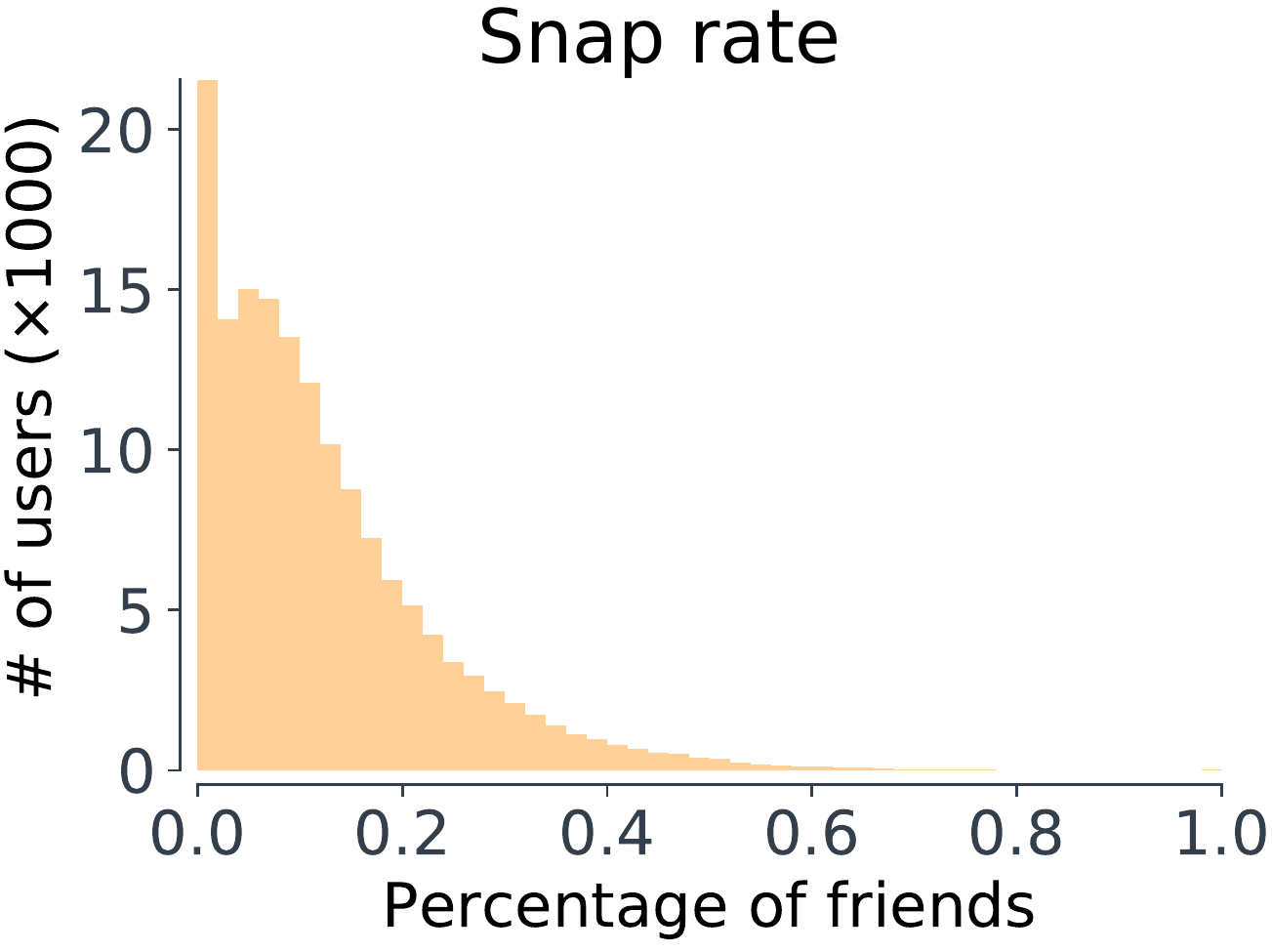}
        \end{subfigure}
    \vspace{-1em}
    \caption{Most users communicate frequently only with a subset (${\leq}20\%$) of their friends, making careful aggregation important when considering influence from neighbors.}
    \label{fig:friendfeq}
    \vspace{-2em}
\end{figure}

The learned node embedding vectors from the \tgcn can be aggregated as a representation for the graph, such as using mean-pooling to average embedding vectors on all nodes.
However, there is a significant disadvantage to such simple solution: namely, the closeness of friends is ignored.
In reality, most users only have a few close friends; users with many friends may only frequently engage with one or few of them.  To validate, we compute the friend communication rate of all \snap users from a selected city (obscured for privacy reasons). Specifically, we compute the percentage of friends that a user has directly communicated (Chat/Snap) with at least once in a two-week span. As Figure \ref{fig:friendfeq} shows, most users mainly communicate with a small percentage (10-20\%) of their friends, and don't frequently contact the remaining ones.
Therefore, friendship activeness is key in precisely modeling the closeness of users.
To this end, we propose a friendship attention mechanism \cite{vaswani2017attention} to quantify the importance of each friend. Formally, a normalized attention score is assigned for each friend $v\in \N(u)$:
\begin{equation} \label{eqn:fnd_attn}
\small
    \alpha_v = \frac{\text{exp} \left( \phi \left( \tx^v \oplus \e^v \right) \right)}{\sum_{\nu  \in \N(u)} \text{exp}\left( \phi \left(\tx^\nu \oplus \e^\nu\right)\right)},
\end{equation}
where  $\tilde{\x}^v$ is the embedding vector of node $v$ from the tensor-based GCN, $\e^v$ is the edge feature on edge between $u$ and $v$, $\oplus$ denotes concatenation, and $\phi(\cdot)$ is a mapping function (e.g., a feed-forward neural network).
Both user actions (preserved by node embedding vectors) and user interactions are considered by the friendship attention mechanism.
To obtain graph representations, we first get the averaged embedding from all friend users weighted by the friendship attention score: 
\begin{equation}
\small
    \hat{\x} = \sum_{v\in \N(u)}\alpha_v\tx^v.
\end{equation}
Then we concatenate it with the embedding vectors on node $u$ alongside each feature category to get the graph embedding:
\begin{equation} \label{eqn:graphemb}
\small
    \g^u = \tx^u \oplus \hat{\x} = \left[\tx^u_1 \oplus \hat{\x}_1, \cdots, \tx^u_K \oplus \hat{\x}_K\right],
\end{equation}
as shown in the right part of Figure \ref{fig:tgcn}.
Note that $\tx^u_k \oplus \hat{\x}_k$ is specifically learned from user action $k$, and  $\g^u \in \R^{K \times (2d^\prime)}$ preserves exclusive information for every user action.
Given $\g^u_1, \cdots \g^u_T$ from $T$ historical user graphs, the next step is to capture temporal dynamics using the temporal module.
\vspace{-.5em}

\subsection{Temporal Module}
As user activities and interactions evolve over time, modeling its temporal dynamics is a key factor of an accurate prediction for user engagement.
Inspired by the success of prior studies for modeling sequential behavior data \cite{liu2019characterizing,yang2018know,yao2019revisiting,tang2019joint} with recurrent neural networks, we utilize LSTM \cite{hochreiter1997long} to capture the evolvement of dynamic user graphs. 
Specifically, we adopt \tlstm following \citeauthor{guo2019exploring}\cite{guo2019exploring}. 
Mathematically, the transformation at each layer of the \tlstm is as follows:
\vspace{-.3em}
{\small
\begin{align} \label{eqn:lstm}
    \nonumber \f_t &= \sigma\left( \g^u_t \otimes \mathcal{U}_f + \h_{t-1} \otimes \mathcal{U}_f^\h + \b_f\right), \\
    \nonumber \mathbf{i}_t &=  \sigma\left( \g^u_t \otimes \mathcal{U}_i + \h_{t-1} \otimes \mathcal{U}_i^\h + \b_i\right), \\
    \nonumber \mathbf{o}_t &= \sigma\left( \g^u_t \otimes \mathcal{U}_o + \h_{t-1} \otimes \mathcal{U}_o^\h + \b_o\right), \\
    \nonumber \mathbf{c}_t &= \f_t \odot \mathbf{c}_{t-1} + \mathbf{i}_t \odot \text{tanh}\left( \g^u_t  \otimes \mathcal{U}_c + \h_{t-1} \otimes \mathcal{U}_c^\h + \b_c\right), \\
    \h_t &= \mathbf{o}_t \odot \text{tanh}\left(\mathbf{c}_t\right),
\end{align}
}
\vspace{-0.3em}
where $\odot$ denotes element-wise multiplication, $\mathcal{U}_*$, $\mathcal{U}_*^\h$ and $\b_*$ are parameters.
Similar to \tgcn, \tlstm can also be considered as a set of parallelized LSTMs, where each LSTM is responsible for a specific feature group corresponding to its user action.
Because the input graph embedding vectors $\g_1^u, \cdots, \g_T^u$ to \tlstm are specific to each feature category (user action), \tlstm can capture the exclusive temporal dependencies of each user action separately.
Similar to $\x$, we define the hidden states of \tlstm as $\h_t = [\h_{t,1},\cdots, \h_{t,K}]$ where $\h_{t,k}$ is exclusively learned for user action $k$.
We further use the hidden states to generate the engagement scores.

\subsection{User Engagement Score Generation}
As aforementioned, user action, temporal dynamics, and friendship are key factors to characterize and predict user engagement.
We introduce three latent variables as $z_A$, $z_J$, $z_I$ to represent different user actions (feature category), time intervals, and friends, respectively so that we can distinguish the influence of specific actions, time intervals, and friends.
For example, different friends may contribute unequally to user engagement; and certain in-App functions could have higher contributions.
Introducing latent variables also bridges the gap between learning explanations and predicting engagement. The desired explanations are importance vectors that constrain the posteriors of latent variables, and further govern the generating of user engagement scores (introduced in Section \ref{sce:explain}).
Specifically, \ours generates user engagement predictions as follows:
{\small
\begin{align}  \label{eqn:gen}
    \nonumber & p\left( e_T|\{\bg_*\}\right) = \sum_{k = 1}^{K} \sum_{t = 1}^{T} \sum_{v=1}^{|\N(u)|} p\left(e_T, z_A = k, z_J = t, z_I = v | \{\bg_*\}\right) \\
    \nonumber & = \sum_{k = 1}^{K} \sum_{t = 1}^{T} \sum_{v=1}^{|\N(u)|} \underbrace{p\left(e_T|z_A = k, z_J = t, z_I = v; \tilde{\x}^v\right)}_{\text{node embedding}} \\
    \nonumber& \cdot \underbrace{p\left(z_I = v | z_J = t, z_A = k, \bg_t\right)}_{\text{friendship attention}}
     \cdot \underbrace{p\left(z_J = t |z_A = k, \{\h_{*,k}\}\right)}_{\text{temporal attention}} \\
     & \underbrace{\cdot p\left(z_A = k |\{\h_*\}\right)}_{\text{user action attention}},
\end{align}
}%
where $\{\h_*\}$ denotes $\{\h_1 \ldots \h_T\}$, and $\{\h_{*,k}\}$ denotes $\{\h_{1, k} \ldots  \\ \h_{T, k}\}$.
The joint probability distribution is further estimated using the conditional probability of latent variables $z_I$, $z_J$, $z_A$, which characterize how user engagement scores are affected by the friendship, temporal dynamics, and user actions accordingly.
We keep designing \ours in accordance with the generation process in Eqn. \ref{eqn:gen}.
In particular, node embeddings are first computed exclusively for every friend, time interval, and user action with proposed \tgcn.
Next, friendship attention $p(z_I = v | z_J = t, z_A = k, \bg_t)$ is estimated using Eqn. \ref{eqn:fnd_attn}.
The summation over $v$ in Eqn. \ref{eqn:gen} is derived by graph representations from friendship modules.
Then \tlstm encapsulates temporal dynamics of graph representation.
The conditional probability of $z_J$ is given as a temporal attention over $\{\h_{*,k}\}$:
{\small
\begin{equation}
    \beta_{t,k} = p\left(z_J = t |z_A = k, \{\h_{*,k}\}\right) = \frac{\text{exp}\left(\varphi_k\left(\h_{t,k}\right)\right)}{\sum_{\tau = 1}^{T}\text{exp}\left(\varphi_k\left(\h_{\tau,k}\right)\right)},
\end{equation}
}%
where $\varphi_k(\cdot)$ is a neural network function specified for user action type $k$. 
Using temporal attention, each user action is represented by its exclusive summarization over all past time intervals as 
\begin{equation}
\small
    \a_k = \sum_{t = 1}^{T} \beta_{t,k}\h_{t,k}
\end{equation}
Finally, we approximate $p(z_A = k |\{\h_*\})$ as the user action attention with another softmax function:
\begin{equation}
\small
    p\left(z_A = k|\{\h_*\}\right) = \frac{\text{exp}\left(\phi\left(\a_k \oplus \h_{T,k}\right)\right)}{\sum_{\kappa = 1}^{K}\text{exp}\left(\phi\left(\a_\kappa \oplus \h_{T,\kappa}\right)\right)},
\end{equation}
where $\phi(\cdot)$ is parameterized by a neural network. 

To approximate the summation over all time intervals ($t=1,\cdots,T$) in Eqn. \ref{eqn:gen},
we use Gaussian distributions to estimate the contribution of every user action to user engagement. Specifically, we use $N(\mu_k, sd_k) = \psi_k(\a_k \oplus \h_{T,k})$ to parameterize the Gaussian distribution for user action $k$. Here $\psi_k(\cdot)$ is also a neural network.
By integrating over all user actions, the user engagement score is derived as:
\begin{equation}
\small
    p\left(e_T\right) = \sum_{k=1}^K N\left(\mu_k, sd_k\right) \cdot p\left(z_A = k|\{\h_*\}\right).
\end{equation}

\subsection{Explainable User Engagement} \label{sce:explain}
To interpret the predicted user engagement, \ours learns the importance vectors
as explanations.
Similar to many previous studies (e.g., \cite{qin2017dual,choi2018fine,xu2018raim,guo2019exploring}),
the local explanations for individual users are directly derived from proposed mixture attentions.
Specifically, the friendship attention, temporal attention and user action attention are acquired as importance vectors for friendship, temporal and user action, respectively.
Because the computation of these attention scores are included by \ours, it takes no extra cost to derive local explanations.
Local explanations reflect specific characteristics and preferences for individual users, which can change dynamically for certain users.

However, local explanations could only help us understand user engagement from individual level. Taking user action as an example, the distribution of its importance vector could vary a lot among different users (see experiments in Section \ref{exp:act_imp} as an example). Because some functions of the App cannot be personalized for every user, it is necessary to interpret their contributions from a global view. For example, when distributing a new feature in an A/B test, it is more reasonable to understand the impact of the feature globally.
Under such circumstances, we formulate the global interpretation of user engagement as a learning problem, where the global importance vectors are jointly learned with the model.
Taking the global importance vector for user action $\A^*$ as an example, we adopt the Expectation–Maximization (EM) method to learn $\A^*$ jointly with the optimization of model parameters $\theta$: 
{\small
\begin{align} \label{eqn:em}
    \nonumber \L(\theta, \A^*) =  & - \sum_{u\in\S} \mathop{\mathbb{E}}_{q_{A}^u} \left[\text{log}\ p\left(e_T^{u}|z_A^{u}; \{\bg_*^u\}\right)\right] \\ 
   &  - \mathop{\mathbb{E}}_{q_{A}^u} \left[\text{log}\ p\left(z_A^{u}|\{\h_*^{u}\}\right)\right] - \mathop{\mathbb{E}}_{q_{A}^u} \left[\text{log}\ p\left(z_A^{u}|\A^*\right)\right],
\end{align}
}%
where the summation $\sum$ is applied over all training samples $\S$, and $q_{A}^u$ denotes the posterior distribution for ${z_A}^u$:
{\small
\begin{align}
    \nonumber q_{A}^u & =  p\left(z_A^{u}|\{\G_*^{u}\}, e_T^{u}, \theta\right) \propto p\left(e_T^{u}|z_A^{u}, \{\G_*^{u}\}\right) \cdot p\left(z_A^{u}|\{\G_*^{u}\}\right) \\
    & \approx p\left(e_T^{u}|z_A^{u}, \q_k^{u} \oplus \h_{T,k}^{u}\right) \cdot p\left(z_A^{u}|\{\h_*^{u}\}\right).
\end{align}
}%
The last term in Eqn. \ref{eqn:em} serves as a regularization term over the posterior of $z_A^{u}$. Note that the posterior of $z_A^{u}$ governs the user action attention. Consequently, the regularization term  encourages the action importance vectors of individual users to follow the global pattern parameterized by $\A^*$.
Moreover, we can derive the following closed-form solution of $\A^*$ as:
\begin{equation}
\small
    \A^* = \frac{1}{|\S|}\sum_{u\in\S} q_{A}^u,
\end{equation}
which takes both user action attention and the prediction of user engagement into consideration. The learning of user action importance relies on the estimation of posterior $q_{A}^u$.
During training stage, network parameters $\theta$ and the posterior $q_{A}^u$ are estimated alternatively.
Namely, we first freeze all parameters $\theta$ to evaluate $q_{A}^u$ over the batch of samples, then use the updated $q_{A}^u$ with gradient descent to update $\theta$ by minimizing \ref{eqn:em}.
Similarly for the global temporal importance, we derive the following closed-form solution:
\begin{equation}
\small
    \bt^*_{t,k} = \frac{1}{|\S|}\sum_{u\in\S} \beta_{t,k}.
\end{equation}

\subsection{Complexity Analysis}
The proposed \tgcn and adopted \tlstm \cite{guo2019exploring} are more efficient than their vanilla versions. Specifically, we have:
\begin{theorem} \label{theo:efficient}
Let $d_{\text{in}}$ and $d_{\text{out}}$ denote input and output dimensions of a layer. The tensor-based designs for GCN and LSTM reduce network complexity by $(1-1/K)d_{\text{in}} \cdot d_{\text{out}}$ and $4(1-1/K)(d_{\text{in}} + d_{\text{out}})d_{\text{out}}$ trainable parameters, and reduce the computational complexity by $\mathcal{O}\left(d_{\text{in}} \cdot d_{\text{out}}\right)$ and $\mathcal{O}\left(\left(d_{\text{in}} + d_{\text{out}}\right)d_{\text{out}}\right)$, respectively.
\end{theorem}
\begin{proof}
\vspace{-5pt}
We provide the proof in Appendix \ref{apd:complexity}.
\end{proof}

As a result, the proposed designs accelerate the training and inference of \ours, and produce a more compact model.  Appendix \ref{apd:runtime_experiments} shows that \ours's tensor-based design reduces training and inference time by ${\approx}20\%$ compared to using the vanilla version (GCN/LSTM).

\section{Evaluation}
In this section, we aim to answer the following research questions: 
\begin{itemize}[leftmargin=*]
\item \textbf{RQ1}:  Can \ours outperform state-of-the-art alternatives in the user engagement prediction task? 
\item \textbf{RQ2}: How does each part/module in \ours affect performance?
\item \textbf{RQ3}: Can \ours derive meaningful explanations for friendships, user actions, and temporal dynamics?
\item \textbf{RQ4}: Can \ours flexibly model different engagement metrics?
\end{itemize}

\subsection{Datasets and Experiment Setup}
We obtain two large-scale datasets from \snap. Each dataset is constructed from all users that live in a different city (on two different continents), we filter out inactive/already churned users.  We follow previous studies on \snap \cite{liu2019characterizing} and collect 13 representative features for user actions on \snap, normalizing to zero mean and unit variance independently before training. 
Table \ref{tab:action} in Appendix provides explains each feature.
We consider 1-day time intervals over 6 weeks. We use the 3 weeks for training, and the rest for testing. We use 2 weeks of user graphs as input to predict engagement in the following week (i.e., $\Delta t = 7d$).

To show that \ours is general for multiple prediction scenarios, we evaluate on two notions of user engagement.  The first metric considers user session time in hours (winsorized to remove extreme outliers). 
The second metric considers \textit{snap} related activities, which are core functions of \snap. We aggregate and average four normalized snap related features, including send, view, create and save, as the measurement for user engagement. 
The prediction of user engagement scores based on two different metrics is denoted by \taskone and \tasktwo, respectively.
We choose root mean square error (RMSE), mean absolute percentage error (MAPE), and mean absolute error (MAE) as our evaluation metrics.  We run all experiments 10 times and report the averaged results. Other technical details are discussed in Appendix \ref{imp_details}. Our code is publicly available on \textbf{Github}\footnote{https://github.com/tangxianfeng/FATE}.

\subsection{Compared Methods}
To validate the accuracy of user engagement prediction,
we compare \ours with the following state-of-the-art methods:
\begin{itemize}[leftmargin=*]
    \item \textbf{Linear Regression} (LR): we utilize the averaged feature vectors of each node in $\bg_t$ as a representation for time interval $t$, and concatenate the vectors over all past time intervals as the input.
    \item \textbf{XGBoost} (XGB) \cite{chen2016xgboost}: We adopt the same prepossessing steps of LR as input for XGBoost.
    \item \textbf{MLP} \cite{hastie2009elements}: We experiment on a two-layer MLP  with the same input features to LR and XGBoost.
    \item \textbf{LSTM} \cite{hochreiter1997long}: LSTM is a popular RNN model for various sequential prediction tasks. We implement a two-layer LSTM which iterates over historical user action features. The final output is fed into a fully-connected layer to generate prediction.
    \item \textbf{GCN} \cite{kipf2016semi}: We combine all historical dynamic friendship graphs into a single graph.
    For each user, we concatenate action features over the observed time period into a new nodal feature vector. 
    \item \textbf{Temporal GCN-LSTM} (TGLSTM) \cite{liu2019characterizing}: TGLSTM is designed to predict future engagement of users, and can be treated as current state-of-the-art baseline. TGLSTM first applies GCN on action graph at each time interval, then leverage LSTM to capture temporal dynamics. We adopt the same design following \citeauthor{liu2019characterizing}\cite{liu2019characterizing} and train TGLSTM on our data to predict the engagement score.
\end{itemize}

To measure the explainability of \ours, we compare with the feature importance of XGB, and LSTM with temporal attention. 
After the boosted trees of XGB are constructed, the importance scores for input features are retrieved and reshaped as an explanation for temporal importance.  For LSTM, we compute attention scores over all hidden states  as an explanation for time intervals.


\subsection{User Engagement Prediction Performance}

\definecolor{LightCyan}{rgb}{0.88,1,1}
\begin{table}[t]
\small
\setlength{\tabcolsep}{0.5pt}
\renewcommand{\arraystretch}{0.9}
\caption{\ours consistently outperforms alternative models in prediction error metrics on both \taskone and \tasktwo, and both datasets \regionone and \regiontwo.}
\vskip -1em
\begin{tabular}{cccccccc}
\toprule
\multicolumn{2}{c}{\multirow{2}{*}{}}  & \multicolumn{3}{c}{\regionone} & \multicolumn{3}{c}{\regiontwo} \\
\cmidrule(lr){3-5}
\cmidrule(lr){6-8}
\multicolumn{2}{c}{}                          & \textbf{RMSE}     & \textbf{MAPE}     & \textbf{MAE}     & \textbf{RMSE}     & \textbf{MAPE}     & \textbf{MAE}   \\ \midrule
\multirow{7}{*}{\rotatebox[origin=c]{90}{\taskone}} & \textbf{LR}                   & .188$\pm$.001    & .443$\pm$.001    & .153$\pm$.000   & .183$\pm$.000    & .375$\pm$.001    & .151$\pm$.000   \\
                        & \textbf{XGB}              & .141$\pm$.000    & .260$\pm$.000    & .101$\pm$.000   & .140$\pm$.000    & .224$\pm$.001    & .098$\pm$.000   \\
                        & \textbf{MLP}                  & .139$\pm$.003    & .233$\pm$.007    & .094$\pm$.004   & .125$\pm$.005    & .238$\pm$.011    & .095$\pm$.004   \\
                        & \textbf{GCN}                  & .131$\pm$.012    & .228$\pm$.019    & .094$\pm$.007   & .128$\pm$.008    & .242$\pm$.010    & .101$\pm$.003   \\
                        & \textbf{LSTM}                 & .121$\pm$.005   & .221$\pm$.003    & .093$\pm$.003   & .122$\pm$.002    & .213$\pm$.005    & .095$\pm$.004   \\
                        & \textbf{TGLSTM}             & .114$\pm$.002    & .215$\pm$.005    & .088$\pm$.000   & .122$\pm$.005    & .201$\pm$.004    & .093$\pm$.002   \\
                        
                        & \textbf{\ours} & \textbf{.109$\pm$.003}   & \textbf{.204$\pm$.001}   & \textbf{.081$\pm$.001}  & \textbf{.118$\pm$.002}   & \textbf{.196$\pm$.003}   & \textbf{.088$\pm$.000}   \\ \midrule
\multirow{7}{*}{\rotatebox[origin=c]{90}{\tasktwo}} & \textbf{LR}                  & .201$\pm$.000    & .674$\pm$.001    & .160$\pm$.000   & .190$\pm$.000   & .553$\pm$.000    & .151$\pm$.000   \\
                        & \textbf{XGB}              & .100$\pm$.000    & .347$\pm$.000    & .078$\pm$.001   & .134$\pm$.000    & .337$\pm$.000    & .089$\pm$.001   \\
                        & \textbf{MLP}                  & .088$\pm$.003    & .288$\pm$.006    & .066$\pm$.003   & .101$\pm$.002    & .261$\pm$.005    & .075$\pm$.000   \\
                        & \textbf{GCN}                 & .094$\pm$.006    & .294$\pm$.008    & .069$\pm$.004   & .100$\pm$.002    & .257$\pm$.013    & .072$\pm$.003   \\
                        & \textbf{LSTM}                 & .080$\pm$.002    & .249$\pm$.005    & .059$\pm$.002   & .097$\pm$.002    & .235$\pm$.003    & .070$\pm$.002   \\
                        & \textbf{TGLSTM}             & .079$\pm$.001    & .241$\pm$.006    & .058$\pm$.000   & .095$\pm$.001    & .239$\pm$.003    & .070$\pm$.001   \\
                        
                        & \textbf{\ours}  & \textbf{.072$\pm$.001}   & \textbf{.213$\pm$.003}   & \textbf{.053$\pm$.000}  & \textbf{.093$\pm$.000}   & \textbf{.224$\pm$.002}   & \textbf{.066$\pm$.000}   \\ \bottomrule 
\end{tabular}
    \label{tab:acc}
\vspace{-1.8em}
\end{table}

To answer the first research question, we report user engagement prediction accuracy of above methods in Table \ref{tab:acc}. As we can see, \ours achieves best performance in both tasks. As expected, \ours significantly out-performs two feature-based methods LR and XGB since it captures  friendship relation and temporal dynamics. Deep-learning based methods MLP, GCN, and LSTM achieves similar performance. However, \ours surpasses them with tremendous error reduction. Moreover, \ours outperforms state-of-the-art approach TGLSTM, by at most 10\%. There are two potential reasons. First, \ours additionally captures friendship relation by explicitly modeling user-user interaction. Secondly, \tgcn and \tlstm maintain independent parameters to capture exclusive information for every user actions, which enhances the predicting accuracy.

\subsection{Ablation Study}
To answer the second question, we design four variations of \ours as follow:
\textbf{(1)} $\text{\ours}_{ts}$: We first evaluate the contribution of tensor-based design. To this end, we employ the original GCN \cite{kipf2016semi} and LSTM \cite{hochreiter1997long} to create the first ablation $\text{\ours}_{ts}$. We use the last output from LSTM to predict user engagement score.
\textbf{(2)} $\text{\ours}_{fnd}$: We then study the effectiveness of the friendship module. We apply \tlstm on raw features to create $\text{\ours}_{fnd}$.
\textbf{(3)} $\text{\ours}_{tmp}$: Next we study the contribution from the temporal module. $\text{\ours}_{tmp}$ first concatenate outputs from all friendship modules, then apply a fully-connected layer to generate user engagement score.
\textbf{(4)} $\text{\ours}_{int}$: To analyze the contribution of explicitly modeling user interactions, we remove this part to create the last ablation $\text{\ours}_{int}$.  The performance of all variations are reported in Table \ref{tab:abla_acc}. 
$\text{\ours}_{ts}$ performs worse when compared to \ours because it fails to extract exclusive information from each user action. However, it still outperforms TGLSTM, since user interactions enhance the modeling of friendship relation. 
The comparisons among $\text{\ours}_{fnd}$, $\text{\ours}_{tmp}$ and \ours indicate the effectiveness of modeling friendship and temporal dependency for predicting user engagement.
The comparison between $\text{\ours}_{int}$ and \ours highlights the contribution of user interactions, which help \ours filter inactive friends and pinpoint influential users.

\begin{table}[t]
\small
\setlength{\tabcolsep}{0.2pt}
\renewcommand{\arraystretch}{0.9}
\caption{All components help \ours:  Removing (a) \tgcn/\tlstm, (b) friendship module, (c) temporal module or (d) user interactions hurts performance.}
\vspace{-1em}
\begin{tabular}{cccccccc}
\toprule
\multicolumn{2}{c}{\multirow{2}{*}{}}  & \multicolumn{3}{c}{\regionone} & \multicolumn{3}{c}{\regiontwo} \\
\cmidrule(lr){3-5}
\cmidrule(lr){6-8}
\multicolumn{2}{c}{}                          & \textbf{RMSE}     & \textbf{MAPE}     & \textbf{MAE}     & \textbf{RMSE}     & \textbf{MAPE}     & \textbf{MAE}     \\ 
\midrule
\multirow{5}{*}{\rotatebox[origin=c]{90}{\taskone}} & $\textbf{\ours}_{ts}$              &.112$\pm$.002 &	.213$\pm$.004  &	.085$\pm$.001  &	.120$\pm$.000  &	.199$\pm$.001  &	.093$\pm$.000  \\
                        & $\textbf{\ours}_{fnd}$           & .119$\pm$.002  &	.218$\pm$.002  &	.089$\pm$.002  &	.121$\pm$.000  &	.199$\pm$.001  &	.090$\pm$.001 \\
                        & $\textbf{\ours}_{tmp}$              &  .126$\pm$.001  &	.221$\pm$.003  &	.097$\pm$.002  &	.123$\pm$.002  &	.220$\pm$.002  &	.097$\pm$.000   \\
                        & $\textbf{\ours}_{int}$                 & .112$\pm$.001  &	.208$\pm$.001 &	.086$\pm$.002 &	.119$\pm$.002 &	.198$\pm$.002 &	.091$\pm$.000  \\ 
                        & \textbf{\ours} & \textbf{.109$\pm$.003}   & \textbf{.204$\pm$.001}   & \textbf{.081$\pm$.001}  & \textbf{.118$\pm$.002}   & \textbf{.196$\pm$.003}   & \textbf{.088$\pm$.000}   \\ \midrule
\multirow{5}{*}{\rotatebox[origin=c]{90}{\tasktwo}} & 
                            $\textbf{\ours}_{ts}$                   &.078$\pm$.001&	.233$\pm$.004&	.057$\pm$.002&	.095$\pm$.001&	.238$\pm$.003&	.070$\pm$.002   \\
                        & $\textbf{\ours}_{fnd}$  &  .076$\pm$.003&	.228$\pm$.002&	.057$\pm$.002&	.094$\pm$.002&	.231$\pm$.001&	.068$\pm$.000  \\
                        & $\textbf{\ours}_{temp}$                  &.083$\pm$.004	&.240$\pm$.005&	.061$\pm$.002&	.102$\pm$.003&	.253$\pm$.003&	.071$\pm$.001   \\
                        & $\textbf{\ours}_{int}$                  &.075$\pm$.000&	.219$\pm$.001&	.055$\pm$.000&	.094$\pm$.001&	.227$\pm$.002&	.068$\pm$.002    \\ 
                        & \textbf{\ours}  & \textbf{.072$\pm$.001}   & \textbf{.213$\pm$.003}   & \textbf{.053$\pm$.000}  & \textbf{.093$\pm$.000}   & \textbf{.224$\pm$.002}   & \textbf{.066$\pm$.000}   \\ \hline 
\end{tabular}
    \label{tab:abla_acc}
    \vspace{-1.5em}
\end{table}

\subsection{Explainability Evaluation}
To answer the third research question, we first analyze the explanations derived from \ours. Then we compare the results with explanations from baseline methods.

\subsubsection{User Action Importance} \label{exp:act_imp}
We first study the global user action importance $\A^*$. Figure \ref{fig:act_imp} illustrates the importance score of different user actions, where a larger value indicates higher importance for user engagement.

\begin{figure}[t]
        \centering
        \includegraphics[width=0.75\columnwidth]{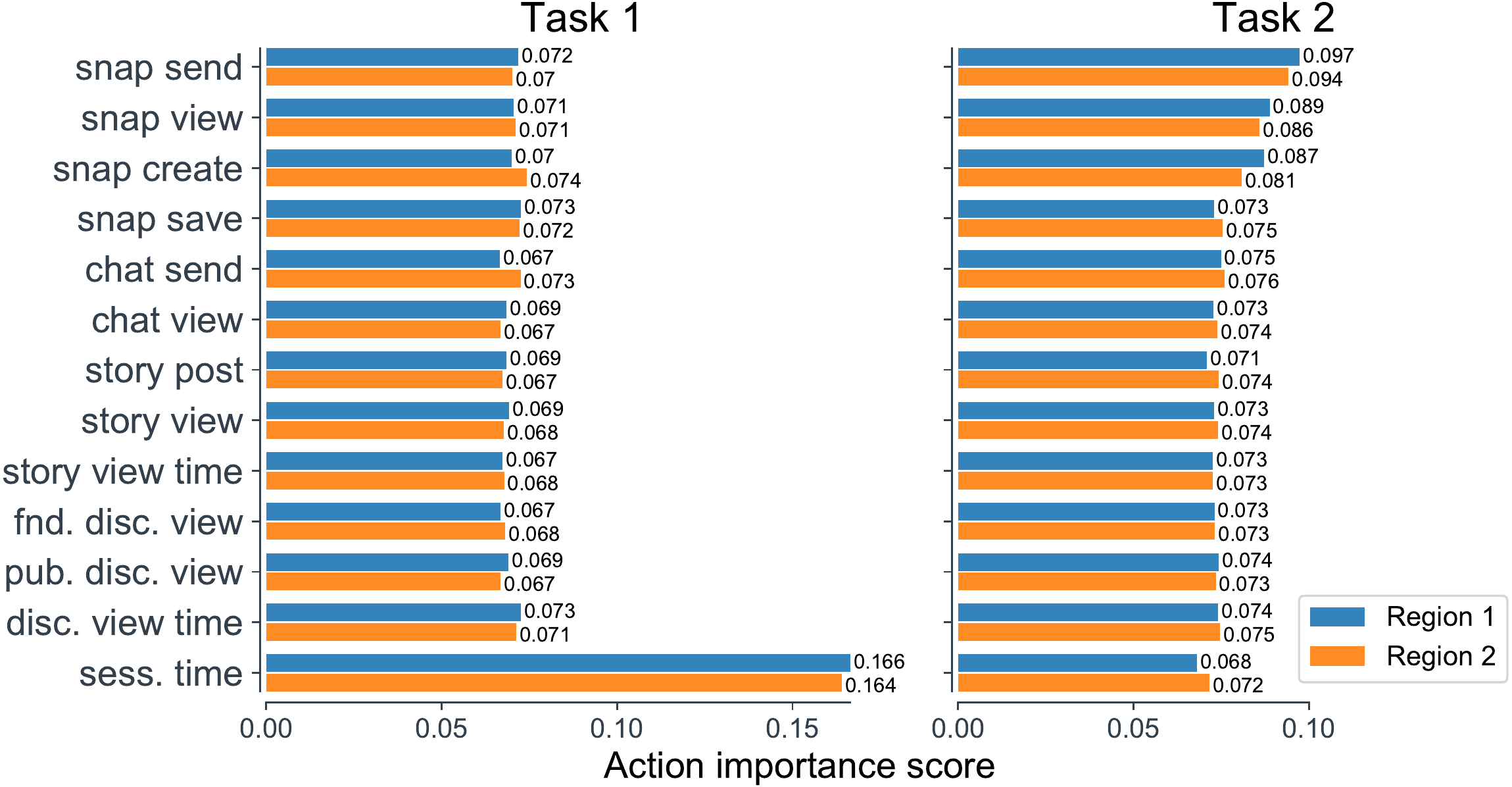}
        \vspace{-1.5em}
        \caption{\ours's global user action importances derived on \taskone and \tasktwo correctly infer past \textit{session time} and \textit{snap}-related actions as the most important for prediction.}
        \vspace{-1.5em}
        \label{fig:act_imp}
\end{figure}

Since the objective of \taskone is to predict a session time-based engagement score, the importance of historical app usage length is significantly higher. This indicates that historical session time is the key factor for user engagement (defined by the expectation of session time in the future), as user activities usually follow strong temporal periodicity.
Remaining user actions play similar roles in extending session time, which is intuitive, because on the entire App level, all the represented in-App functions are heavily consumed.
However, we see that \feature{Snap}-related actions are relatively more important than others. A potential reason is that sending and receiving Snaps (images/videos) are core functions which  distinguish \snap from other  Apps and define product value.

For predicting user engagement defined on normalized \feature{Snap}-related actions in \tasktwo, we see that \feature{SnapSend}, \feature{SnapView}, and \feature{SnapCreate} play the most important role. 
\feature{SnapSend} contributes more to user engagement comparing with \feature{SnapView}, as sending is an active generation activity while viewing is passively receiving information.
Similarly, \feature{SnapCreate} is more important than \feature{SnapSave}, for the reason that creating a Snap is the foundation of many content generation activities, whereas Snap-saving is infrequent. Besides \feature{Snap}-related actions, \feature{ChatSend} is the most important, which makes sense given that private Chat messaging is the next most common usecase after Snaps on \snap, and users often respond to Snaps with Chats and vice-versa.

\begin{figure}[t]
        \centering
        \includegraphics[width=0.7\columnwidth]{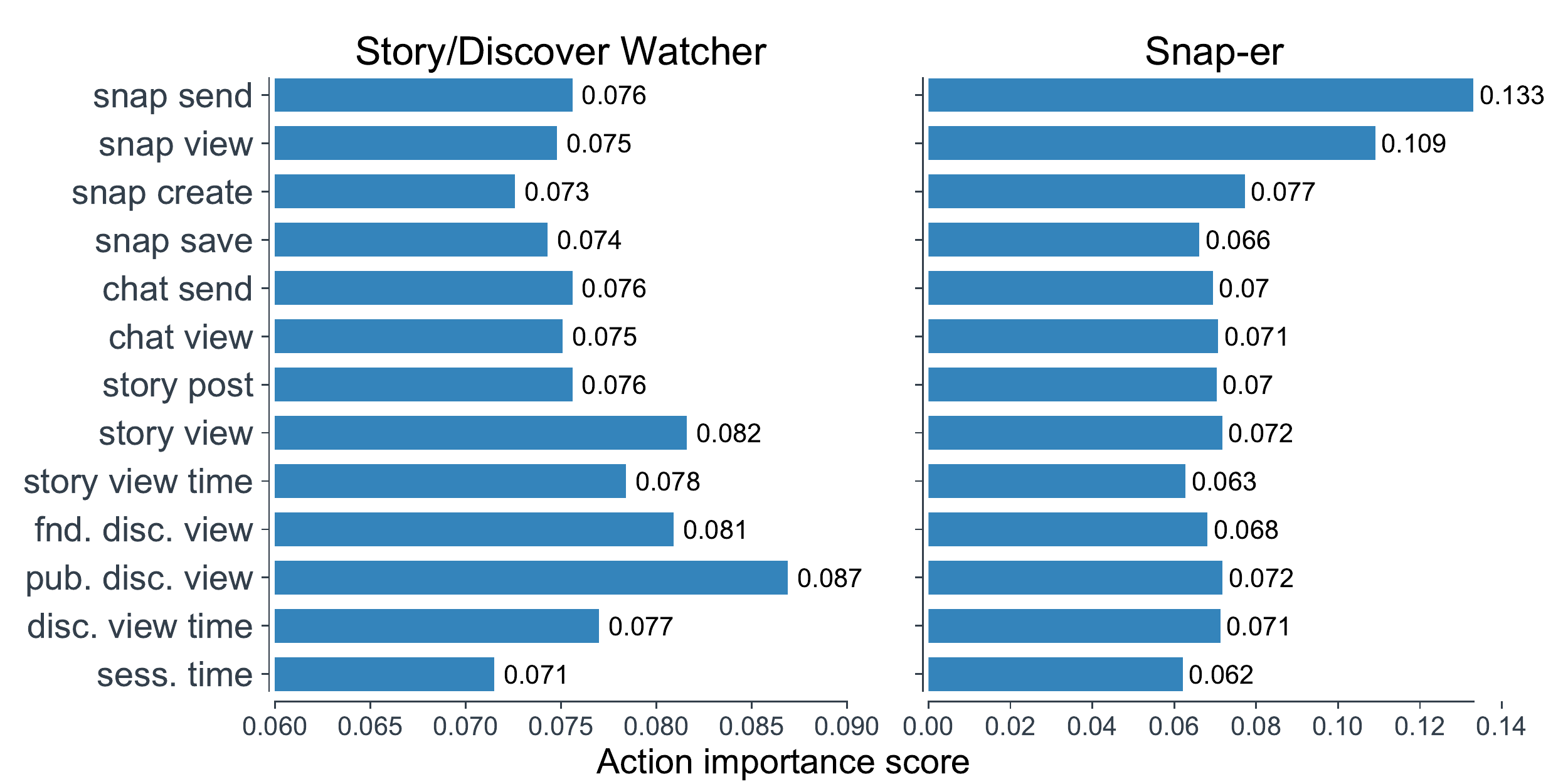}
        \vspace{-1em}
        \caption{Two sample local user action importances: Users with different dominating engagement behaviors exhibit different user action importances.}
        \label{fig:act_imp_ind}
        \vspace{-1.5em}
\end{figure}

Next, we analyze user action importance for individual users. We take \tasktwo as an example, and select two random users from \regionone. To help understand user preference and characteristics, we query an internal labeling service that categorizes users according to their preferences for different \snap features. The service, built on domain knowledge, is as independent from \ours. Generally, a ``Snap-er'' uses \feature{Snap}-related functions more frequently, while a ``Story/Discover Viewer'' is more active on watching friend/publisher Story content on \snap. As illustrated in Figure \ref{fig:act_imp_ind}, the importance scores of \feature{Snap}-related user actions of a Snap-er are significantly higher than that of remained user actions. However, for Story/Discover Viewers, other actions (\feature{StoryView}, \feature{Public-\\DiscoverView}) contribute more. This shows the diversity of action importance for individual users, as the distribution of importance scores changes according to user characteristics.

\subsubsection{Temporal Importance}

\begin{figure}[t]
        \centering
        \begin{subfigure}[b]{.4\textwidth}
            \centering
            \includegraphics[width=\columnwidth]{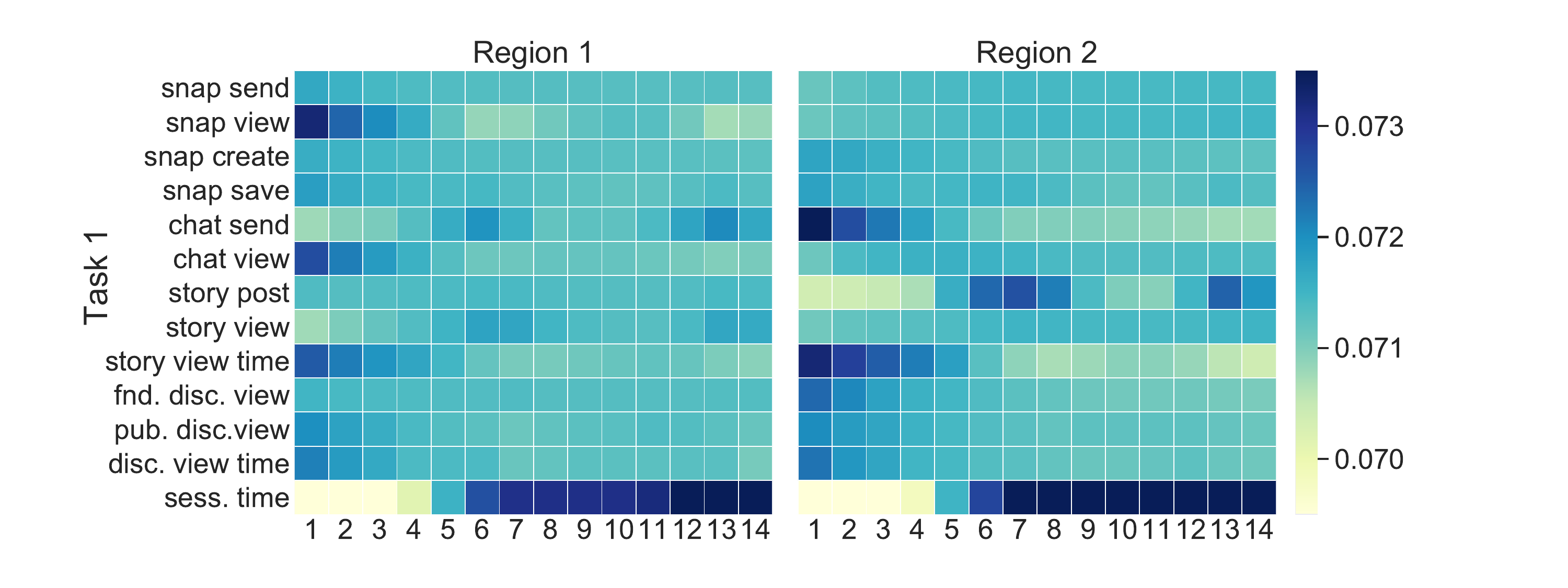}
        \end{subfigure}
        \begin{subfigure}[b]{.4\textwidth}  
            \centering 
            \includegraphics[width=\columnwidth]{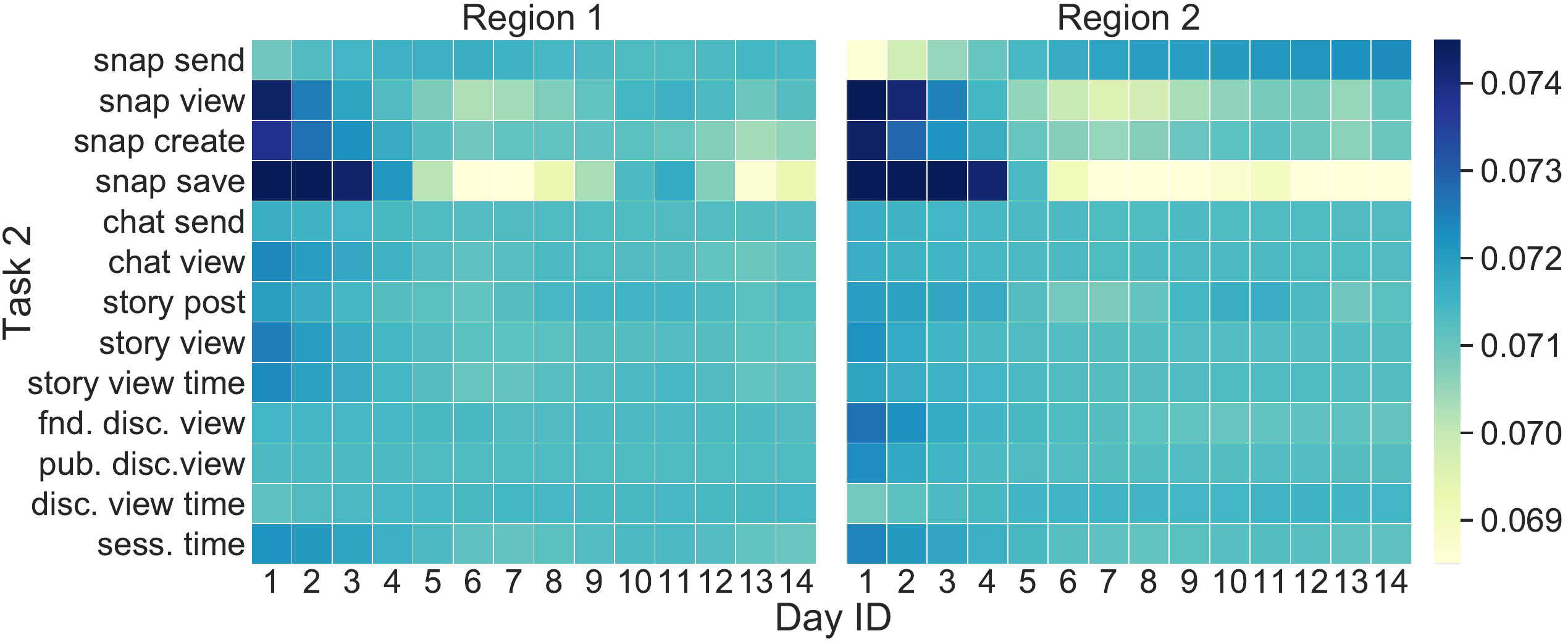}
        \end{subfigure}
        \vspace{-1em}
        \caption{\ours's global temporal importances show long and short-term action importances over time.}
        \label{fig:tmp_imp}
        \vspace{-1.5em}
\end{figure}

Figure \ref{fig:tmp_imp} displays the overall temporal importance of user actions across time (i.e., past 14 days). Darker hue indicates higher importance to user engagement.
For \taskone, \feature{SessionTime} has strong short-term importance in both cities. Temporally close \feature{SessionTime} (later days) data contributes to user engagement more.  On the contrary, other user actions show long-term importance. For example, \feature{SnapView} and \feature{ChatView} show relatively higher importance on the first day. In addition to long/short-term characteristics, we see the importance of most user actions showing strong periodicity in a weekly manner. 
Similar conclusions can also be drawn from \tasktwo, where \feature{SnapView}, \feature{SnapCreate}, and \feature{SnapSave} show longer-term correlation to user engagement. \feature{SnapSend} on the other hand demonstrates a short-term correlation. The periodicity of temporal importance is also relatively weaker compared to \taskone.


We then study the temporal importance for individual users. Similar to action importance, we randomly select two users from \regionone, and plot temporal importance scores when predicting user engagement score in \taskone. As shown in Figure \ref{fig:tmp_imp_ind}, users with different dominant behaviors exhibit different temporal importance score distributions.  The temporal importance scores of \feature{Publisher-\\DiscoverView} and \feature{DiscoverViewTime} are relatively higher for the Story/Discover Watcher, with clear periodicity effects (importance in day 1-2, and then in 7-8, and again in 13-14, which are likely weekends when the user has more time to watch content). The Chatter has higher score for \feature{Chat}-related features, with more weight on recent \feature{ChatView}s (days 12-14).  Our results suggest that explanations learned by \ours coincide with past understanding of temporal influence in these behaviors.

\begin{figure}[t]
    \centering
    \includegraphics[width=.4\textwidth]{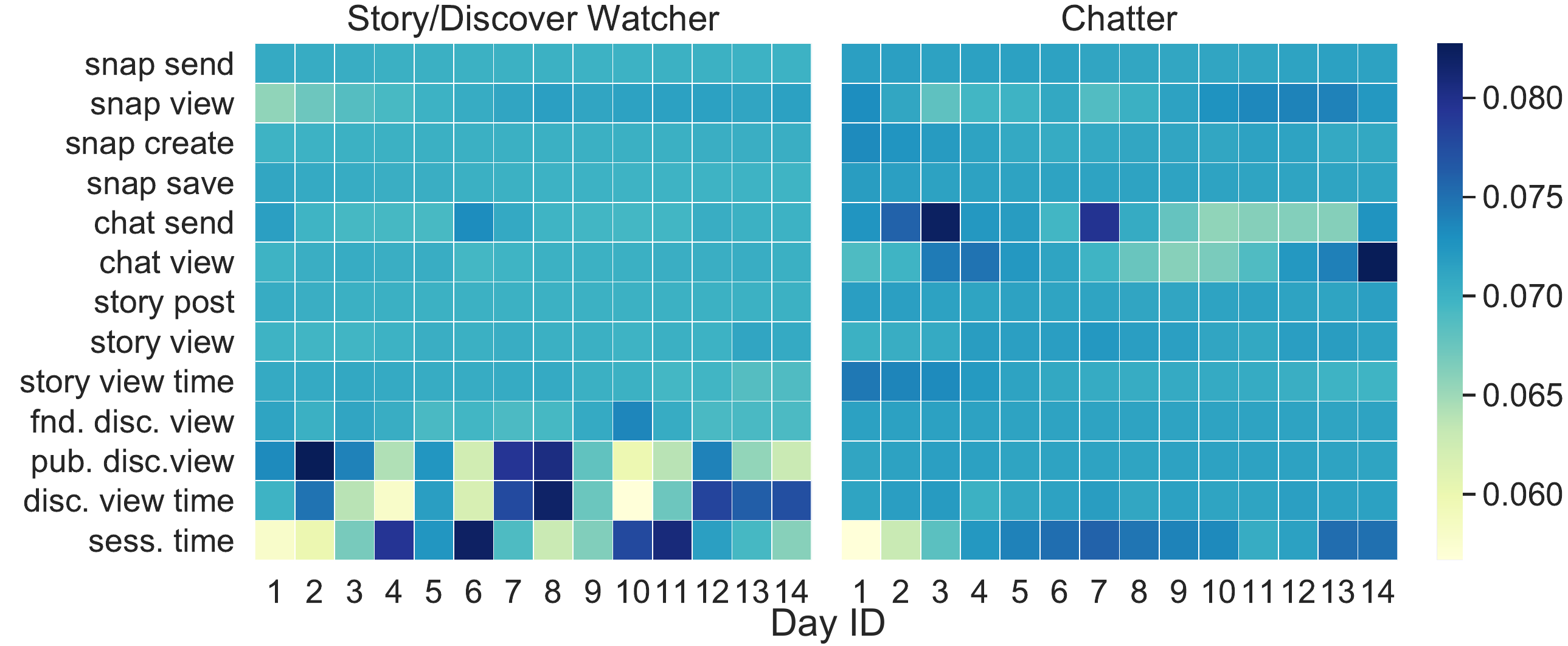}
    \vspace{-1.2em}
    \caption{\ours can capture diverse local level temporal importance for users of different persona.}
    \label{fig:tmp_imp_ind}
    \vspace{-2em}
\end{figure}

\subsubsection{Friendship Importance}
We next validate the learned (local) friendship importance. Figure \ref{fig:fnd_imp} demonstrates two example users selected from \regionone, for \taskone. The heatmaps illustrate the importance scores of their friends.  Clearly, friendship importance scores are not uniformly distributed among all friends. Some friends hold higher importance scores to the selected user, while others have relatively lower scores. This is potentially due to low similarity in user activities, or two friends being independently active (but not jointly interactive).  To verify this assumption and interpret friendship importance scores, we compare user activeness (session time) of the selected user with their most important friends and their least importance friends (measured by the sum of scores over 14 days). As Figure \ref{fig:fnd_imp} shows, the both users follow a pattern similar to their most important friends and unlike the least important ones.  Moreover, temporal importance (darker hue) of the highest-importance friend coincides in the temporal heatmaps (left) and the session time activity plots (right) for both users in (a) and (b).  

\begin{figure*}[t]
       \centering
        \begin{subfigure}[b]{\columnwidth}
            \centering
            \includegraphics[width=0.45\columnwidth]{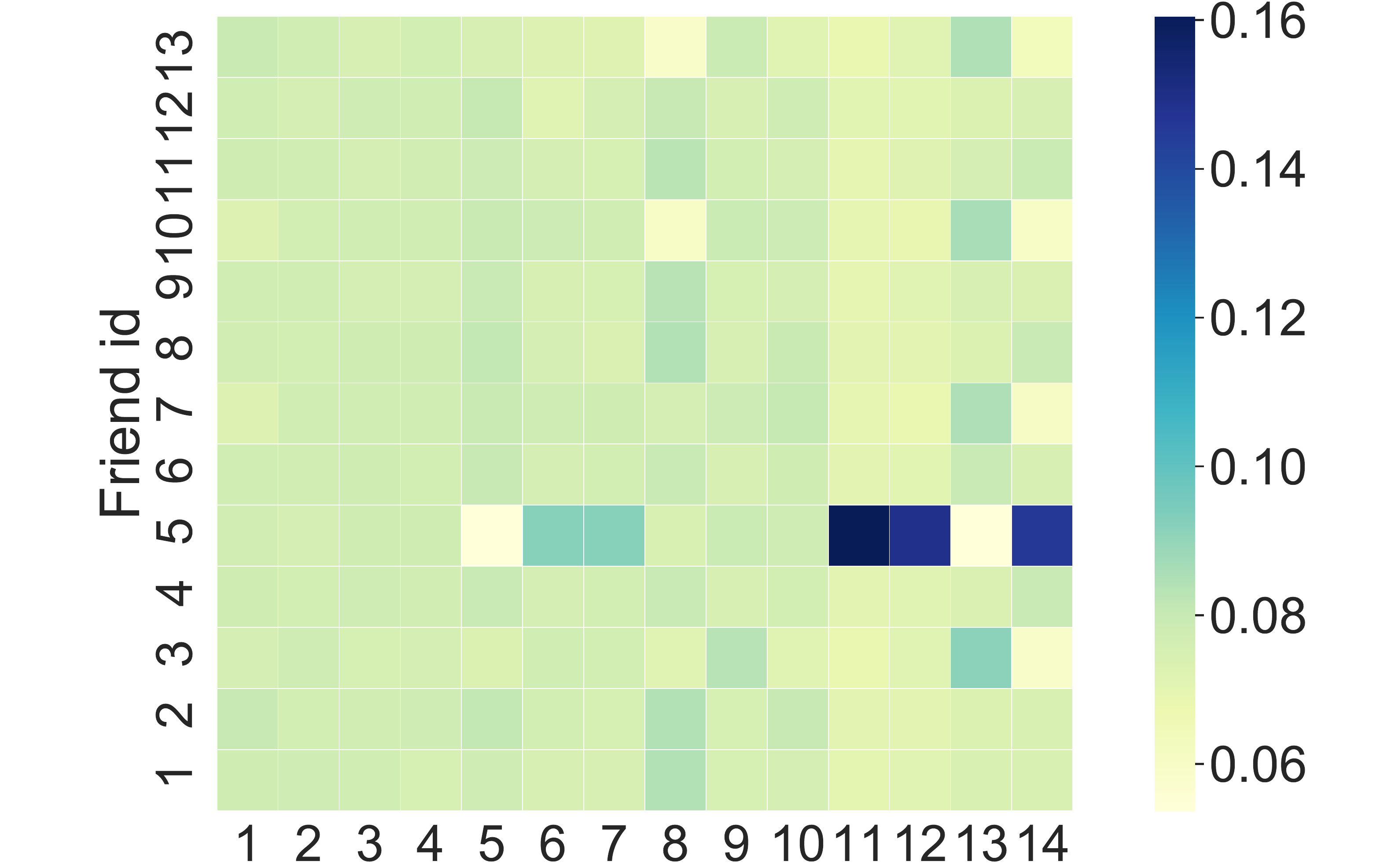}
            \includegraphics[width=0.45\columnwidth]{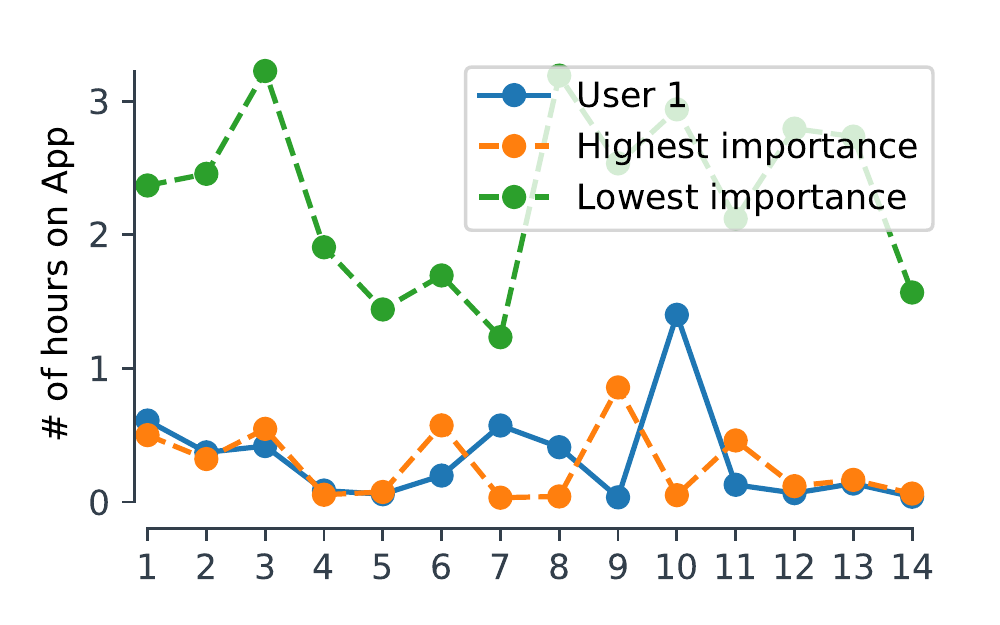}
            \vspace{-0.7em}
            \caption{User 1}
        \end{subfigure}
        \begin{subfigure}[b]{\columnwidth}
            \centering 
            \includegraphics[width=0.45\columnwidth]{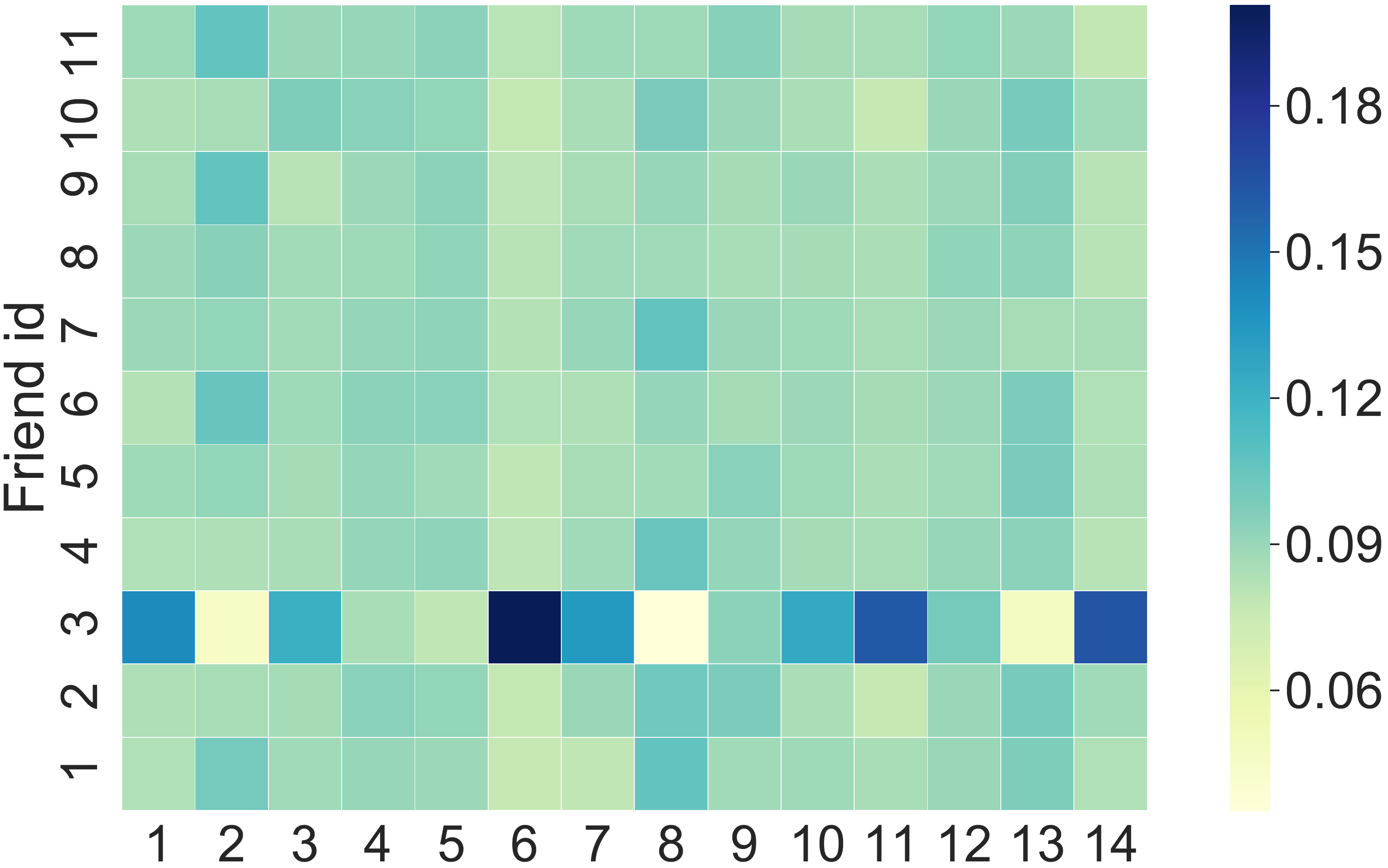}
            \includegraphics[width=0.45\columnwidth]{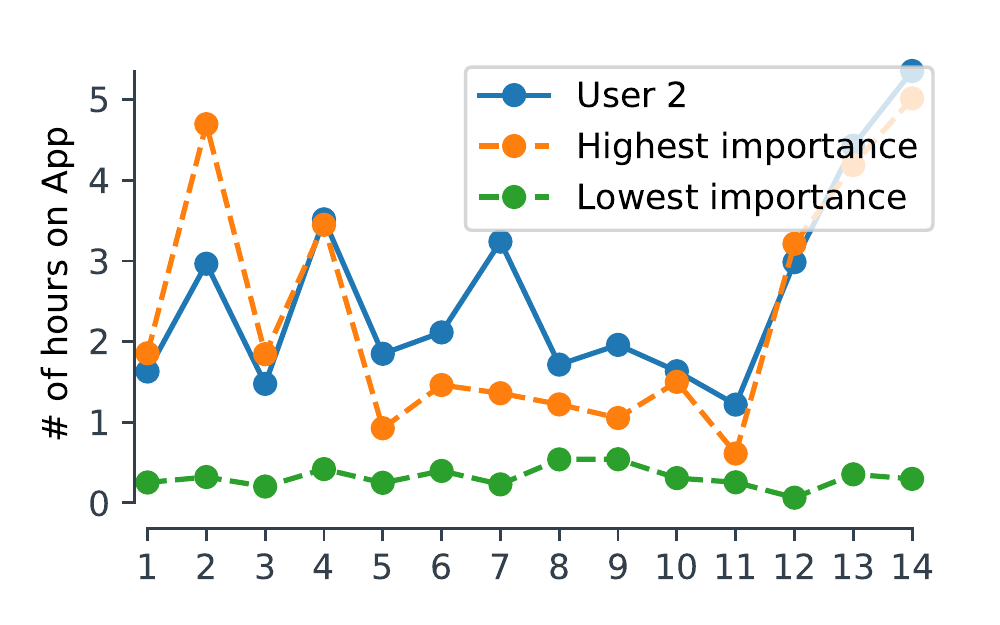}
            \vspace{-0.7em}
            \caption{User 2}
        \end{subfigure}
        \vspace{-1.2em}
        \caption{\ours's local friendship importance captures asymmetric influence of friends: the user has similar session time behaviors (right) as their highest-importance friends (blue and orange lines are close); session time spikes coincide with high temporal importances (left) of those friends (dark hues).}
        \label{fig:fnd_imp}
        \vspace{-1.2em}
\end{figure*}

\begin{figure}[t]
    \centering
    \begin{subfigure}[b]{0.5\columnwidth}
    \includegraphics[width=\columnwidth]{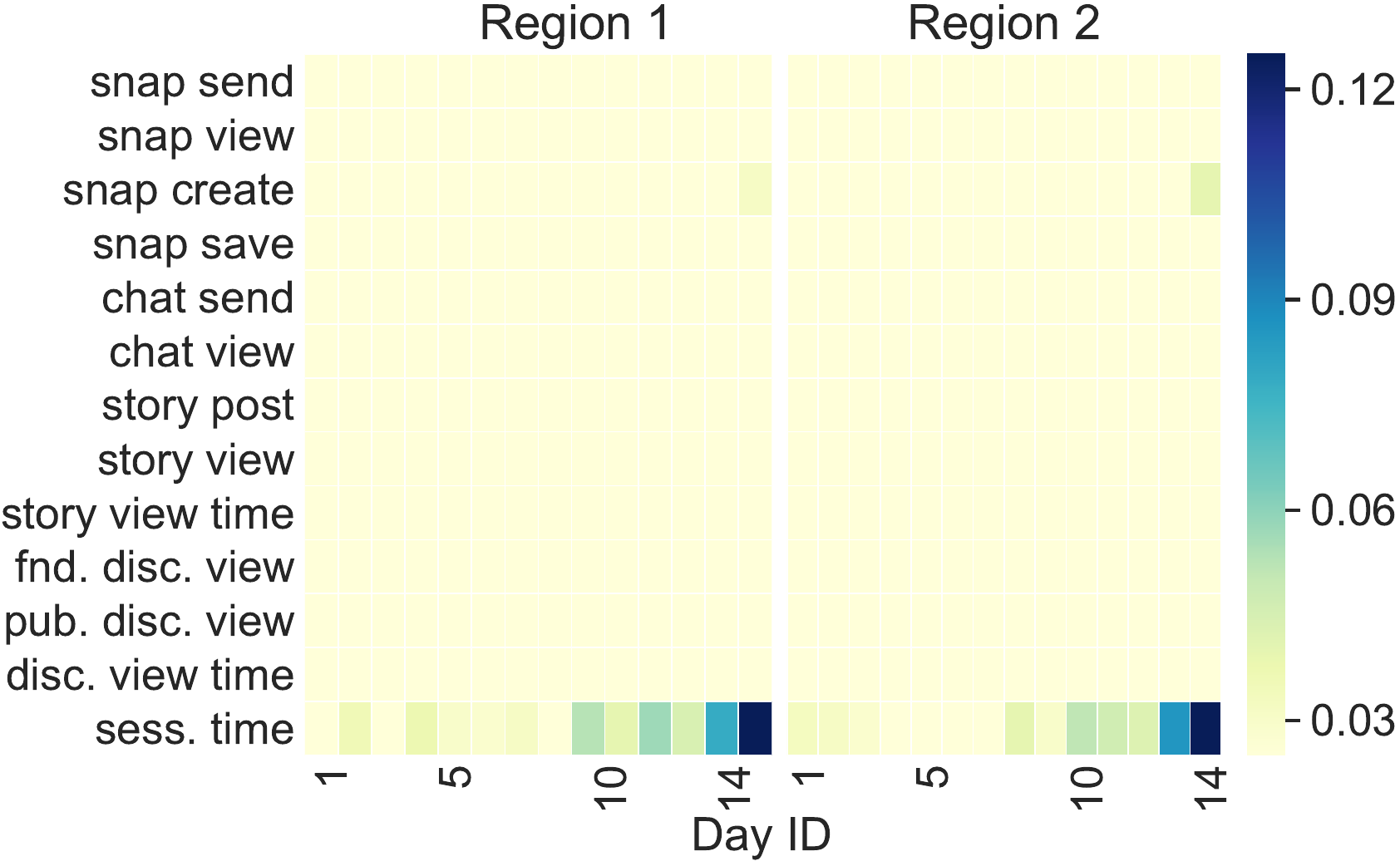}
    \vspace{-2em}
    \caption{XGBoost}
    \end{subfigure}
    \begin{subfigure}[b]{0.3\columnwidth}
    \includegraphics[width=\columnwidth]{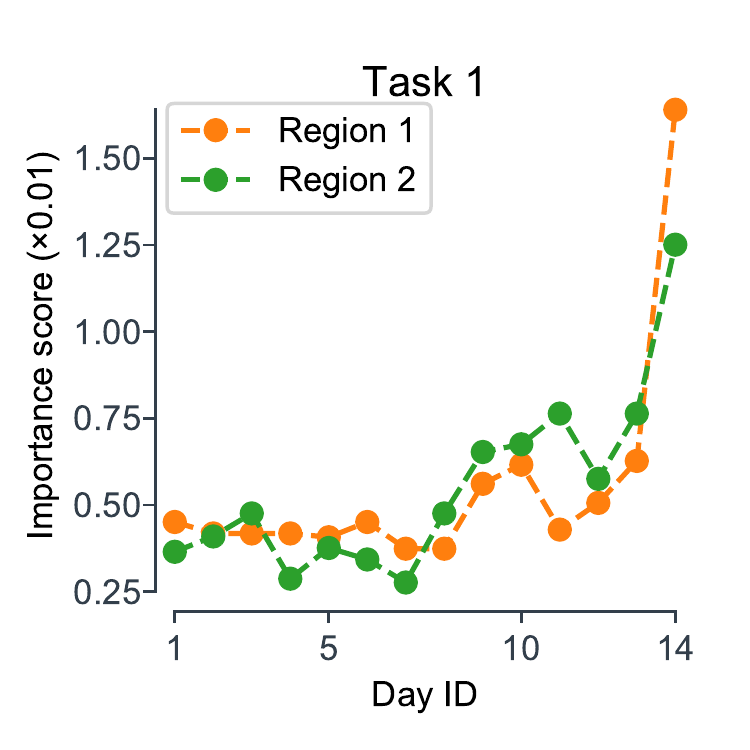}
    \vspace{-2em}
    \caption{LSTM}
    \end{subfigure}
    \vspace{-1.2em}
    \caption{Comparisons of explainability.}
    \label{fig:tmp_imp_xgb}
    \label{fig:tmp_imp_lstm}
    \vspace{-1.8em}
\end{figure}

\subsubsection{Baseline comparisons on explainability}
Feature importance from XGBoost can be used as a temporal importance explanation. As in Figure \ref{fig:tmp_imp_xgb}, results from XGBoost are very sparse, where most user actions receive an unnatural, near-0 importance score, likely because feature importance is only a byproduct of the training of XGBoost. Unlike \ours, the XGBoost objective is purely defined on prediction accuracy, failing to learn  explanations for user actions over time.
Figure \ref{fig:tmp_imp_lstm} shows the temporal attention from LSTM. There are two weakness of using LSTM for explanation: (1) it is unable to capture the importance of each user action; (2) compared to \ours, the temporal attention fails to capture periodicity of user actions, which na\"{i}ve LSTM mixes and cannot separate.  Comparatively, \ours derives richer and more fine-grained explanations.


\subsection{Practical Applications}

Our framework is designed with practical applications in mind.  State-of-the-art in engagement prediction improves temporally-aware estimation of overall demand and key metrics, which offers flexible use in many forecasting and expectation-setting applications.  Explainability in the model helps quantify both global and local factors in user engagement, and how they motivate users to engage with the platform.  Moreover, it paves roads for personalized interventions and nudges to users to realize in-App value, stay in touch with their best friends and retain.  Finally, our choices around tensor-based modeling improve efficiency by reducing parameters and decreasing training time.  Yet, GNN training/inference is still a challenge for multi-million/billion-scale workloads, especially considering dynamism of the underlying data,  temporality of predictions, and frequent model updation needs in practice, though new work in GNN scalability offers some promising inroads \cite{ying2018graph, chiang2019cluster}.  In the future, we plan to develop automated and recurrent training and inference workflows which can handle these issues to gracefully scale \ours to production workloads larger than those we experimented on.


\section{Conclusion}
In this paper, we explore the problem of explainable user engagement prediction for social network Apps.  Given different notions of user engagement, we define it generally as the future expectation of a metric of interest.  We then propose an end-to-end neural framework, \ours, which models friendship, user actions and temporal dynamics, to generate accurate predictions while jointly deriving local and global explanations for these key factors. Extensive experiments on two datasets and two engagement prediction tasks from \snap  demonstrate the efficiency, generality and accuracy of our approach: \ours improves accuracy compared to state-of-the-art methods by ${\approx}10\%$ while reducing runtime by ${\approx}20\%$ owing to its use of proposed tensor-based GCN and LSTM components.  We hope to continue to improve scaling aspects of \ours to deploy it for recurrent auto-training and inference at \snap.  While \ours is designed with \snap in mind, our core ideas of engagement definition, contributing factors, and technical contribution in neural architecture design offer clear applications to other social Apps and online platforms.  

\section*{Acknowledgement}
This material is based upon work supported by, or in part by, the National Science Foundation (NSF) under grant \#1909702. Any opinions, findings, and conclusions in this material are those of the authors and do not reflect the views of the NSF.

\bibliographystyle{ACM-Reference-Format}
\bibliography{ref}


\begin{thebibliography}{43}


\ifx \showCODEN    \undefined \def \showCODEN     #1{\unskip}     \fi
\ifx \showDOI      \undefined \def \showDOI       #1{#1}\fi
\ifx \showISBNx    \undefined \def \showISBNx     #1{\unskip}     \fi
\ifx \showISBNxiii \undefined \def \showISBNxiii  #1{\unskip}     \fi
\ifx \showISSN     \undefined \def \showISSN      #1{\unskip}     \fi
\ifx \showLCCN     \undefined \def \showLCCN      #1{\unskip}     \fi
\ifx \shownote     \undefined \def \shownote      #1{#1}          \fi
\ifx \showarticletitle \undefined \def \showarticletitle #1{#1}   \fi
\ifx \showURL      \undefined \def \showURL       {\relax}        \fi
\providecommand\bibfield[2]{#2}
\providecommand\bibinfo[2]{#2}
\providecommand\natexlab[1]{#1}
\providecommand\showeprint[2][]{arXiv:#2}

\bibitem[\protect\citeauthoryear{Althoff and Leskovec}{Althoff and
  Leskovec}{2015}]%
        {althoff2015donor}
\bibfield{author}{\bibinfo{person}{Tim Althoff} {and} \bibinfo{person}{Jure
  Leskovec}.} \bibinfo{year}{2015}\natexlab{}.
\newblock \showarticletitle{Donor retention in online crowdfunding communities:
  A case study of donorschoose. org}. In \bibinfo{booktitle}{\emph{WWW}}.
  \bibinfo{pages}{34--44}.
\newblock


\bibitem[\protect\citeauthoryear{Au, Chan, and Yao}{Au et~al\mbox{.}}{2003}]%
        {au2003novel}
\bibfield{author}{\bibinfo{person}{Wai-Ho Au}, \bibinfo{person}{Keith~CC Chan},
  {and} \bibinfo{person}{Xin Yao}.} \bibinfo{year}{2003}\natexlab{}.
\newblock \showarticletitle{A novel evolutionary data mining algorithm with
  applications to churn prediction}.
\newblock \bibinfo{journal}{\emph{TEC}} \bibinfo{volume}{7},
  \bibinfo{number}{6} (\bibinfo{year}{2003}), \bibinfo{pages}{532--545}.
\newblock


\bibitem[\protect\citeauthoryear{Benson, Kumar, and Tomkins}{Benson
  et~al\mbox{.}}{2016}]%
        {benson2016modeling}
\bibfield{author}{\bibinfo{person}{Austin~R Benson}, \bibinfo{person}{Ravi
  Kumar}, {and} \bibinfo{person}{Andrew Tomkins}.}
  \bibinfo{year}{2016}\natexlab{}.
\newblock \showarticletitle{Modeling user consumption sequences}. In
  \bibinfo{booktitle}{\emph{WWW}}. \bibinfo{pages}{519--529}.
\newblock


\bibitem[\protect\citeauthoryear{Chen and Guestrin}{Chen and Guestrin}{2016}]%
        {chen2016xgboost}
\bibfield{author}{\bibinfo{person}{Tianqi Chen} {and} \bibinfo{person}{Carlos
  Guestrin}.} \bibinfo{year}{2016}\natexlab{}.
\newblock \showarticletitle{Xgboost: A scalable tree boosting system}. In
  \bibinfo{booktitle}{\emph{KDD}}. ACM, \bibinfo{pages}{785--794}.
\newblock


\bibitem[\protect\citeauthoryear{Chiang, Liu, Si, Li, Bengio, and Hsieh}{Chiang
  et~al\mbox{.}}{2019}]%
        {chiang2019cluster}
\bibfield{author}{\bibinfo{person}{Wei-Lin Chiang}, \bibinfo{person}{Xuanqing
  Liu}, \bibinfo{person}{Si Si}, \bibinfo{person}{Yang Li},
  \bibinfo{person}{Samy Bengio}, {and} \bibinfo{person}{Cho-Jui Hsieh}.}
  \bibinfo{year}{2019}\natexlab{}.
\newblock \showarticletitle{Cluster-gcn: An efficient algorithm for training
  deep and large graph convolutional networks}. In
  \bibinfo{booktitle}{\emph{KDD}}.
\newblock


\bibitem[\protect\citeauthoryear{Choi, Cho, and Bengio}{Choi
  et~al\mbox{.}}{2018}]%
        {choi2018fine}
\bibfield{author}{\bibinfo{person}{Heeyoul Choi}, \bibinfo{person}{Kyunghyun
  Cho}, {and} \bibinfo{person}{Yoshua Bengio}.}
  \bibinfo{year}{2018}\natexlab{}.
\newblock \showarticletitle{Fine-grained attention mechanism for neural machine
  translation}.
\newblock \bibinfo{journal}{\emph{Neurocomputing}}  \bibinfo{volume}{284}
  (\bibinfo{year}{2018}), \bibinfo{pages}{171--176}.
\newblock


\bibitem[\protect\citeauthoryear{Elkahky, Song, and He}{Elkahky
  et~al\mbox{.}}{2015}]%
        {elkahky2015multi}
\bibfield{author}{\bibinfo{person}{Ali~Mamdouh Elkahky}, \bibinfo{person}{Yang
  Song}, {and} \bibinfo{person}{Xiaodong He}.} \bibinfo{year}{2015}\natexlab{}.
\newblock \showarticletitle{A multi-view deep learning approach for cross
  domain user modeling in recommendation systems}. In
  \bibinfo{booktitle}{\emph{WWW}}.
\newblock


\bibitem[\protect\citeauthoryear{Gilpin, Bau, Yuan, Bajwa, Specter, and
  Kagal}{Gilpin et~al\mbox{.}}{2018}]%
        {gilpin2018explaining}
\bibfield{author}{\bibinfo{person}{Leilani~H Gilpin}, \bibinfo{person}{David
  Bau}, \bibinfo{person}{Ben~Z Yuan}, \bibinfo{person}{Ayesha Bajwa},
  \bibinfo{person}{Michael Specter}, {and} \bibinfo{person}{Lalana Kagal}.}
  \bibinfo{year}{2018}\natexlab{}.
\newblock \showarticletitle{Explaining explanations: An overview of
  interpretability of machine learning}. In \bibinfo{booktitle}{\emph{DSAA}}.
  IEEE, \bibinfo{pages}{80--89}.
\newblock


\bibitem[\protect\citeauthoryear{Guo, Lin, and Antulov-Fantulin}{Guo
  et~al\mbox{.}}{2019}]%
        {guo2019exploring}
\bibfield{author}{\bibinfo{person}{Tian Guo}, \bibinfo{person}{Tao Lin}, {and}
  \bibinfo{person}{Nino Antulov-Fantulin}.} \bibinfo{year}{2019}\natexlab{}.
\newblock \showarticletitle{Exploring Interpretable LSTM Neural Networks over
  Multi-Variable Data}.
\newblock \bibinfo{journal}{\emph{arXiv preprint 1905.12034}}
  (\bibinfo{year}{2019}).
\newblock


\bibitem[\protect\citeauthoryear{Hastie, Tibshirani, and Friedman}{Hastie
  et~al\mbox{.}}{2009}]%
        {hastie2009elements}
\bibfield{author}{\bibinfo{person}{Trevor Hastie}, \bibinfo{person}{Robert
  Tibshirani}, {and} \bibinfo{person}{Jerome Friedman}.}
  \bibinfo{year}{2009}\natexlab{}.
\newblock \bibinfo{booktitle}{\emph{The elements of statistical learning: data
  mining, inference, and prediction}}.
\newblock \bibinfo{publisher}{Springer Science \& Business Media}.
\newblock


\bibitem[\protect\citeauthoryear{Hochreiter and Schmidhuber}{Hochreiter and
  Schmidhuber}{1997}]%
        {hochreiter1997long}
\bibfield{author}{\bibinfo{person}{Sepp Hochreiter} {and}
  \bibinfo{person}{J{\"u}rgen Schmidhuber}.} \bibinfo{year}{1997}\natexlab{}.
\newblock \showarticletitle{Long short-term memory}.
\newblock \bibinfo{journal}{\emph{Neural computation}} \bibinfo{volume}{9},
  \bibinfo{number}{8} (\bibinfo{year}{1997}), \bibinfo{pages}{1735--1780}.
\newblock


\bibitem[\protect\citeauthoryear{Jin, Ma, Liu, Tang, Wang, and Tang}{Jin
  et~al\mbox{.}}{2020}]%
        {jin2020graph}
\bibfield{author}{\bibinfo{person}{Wei Jin}, \bibinfo{person}{Yao Ma},
  \bibinfo{person}{Xiaorui Liu}, \bibinfo{person}{Xianfeng Tang},
  \bibinfo{person}{Suhang Wang}, {and} \bibinfo{person}{Jiliang Tang}.}
  \bibinfo{year}{2020}\natexlab{}.
\newblock \showarticletitle{Graph Structure Learning for Robust Graph Neural
  Networks}.
\newblock \bibinfo{journal}{\emph{arXiv preprint arXiv:2005.10203}}
  (\bibinfo{year}{2020}).
\newblock


\bibitem[\protect\citeauthoryear{Kapoor, Sun, Srivastava, and Ye}{Kapoor
  et~al\mbox{.}}{2014}]%
        {kapoor2014hazard}
\bibfield{author}{\bibinfo{person}{Komal Kapoor}, \bibinfo{person}{Mingxuan
  Sun}, \bibinfo{person}{Jaideep Srivastava}, {and} \bibinfo{person}{Tao Ye}.}
  \bibinfo{year}{2014}\natexlab{}.
\newblock \showarticletitle{A hazard based approach to user return time
  prediction}. In \bibinfo{booktitle}{\emph{KDD}}. ACM,
  \bibinfo{pages}{1719--1728}.
\newblock


\bibitem[\protect\citeauthoryear{Kawale, Pal, and Srivastava}{Kawale
  et~al\mbox{.}}{2009}]%
        {kawale2009churn}
\bibfield{author}{\bibinfo{person}{Jaya Kawale}, \bibinfo{person}{Aditya Pal},
  {and} \bibinfo{person}{Jaideep Srivastava}.} \bibinfo{year}{2009}\natexlab{}.
\newblock \showarticletitle{Churn prediction in MMORPGs: A social influence
  based approach}. In \bibinfo{booktitle}{\emph{ICCSE}},
  Vol.~\bibinfo{volume}{4}. IEEE, \bibinfo{pages}{423--428}.
\newblock


\bibitem[\protect\citeauthoryear{Kingma and Ba}{Kingma and Ba}{2014}]%
        {kingma2014adam}
\bibfield{author}{\bibinfo{person}{Diederik~P Kingma} {and}
  \bibinfo{person}{Jimmy Ba}.} \bibinfo{year}{2014}\natexlab{}.
\newblock \showarticletitle{Adam: A method for stochastic optimization}.
\newblock \bibinfo{journal}{\emph{arXiv preprint arXiv:1412.6980}}
  (\bibinfo{year}{2014}).
\newblock


\bibitem[\protect\citeauthoryear{Kipf and Welling}{Kipf and Welling}{2017}]%
        {kipf2016semi}
\bibfield{author}{\bibinfo{person}{Thomas~N Kipf} {and} \bibinfo{person}{Max
  Welling}.} \bibinfo{year}{2017}\natexlab{}.
\newblock \showarticletitle{Semi-Supervised Classification with Graph
  Convolutional Networks}. In \bibinfo{booktitle}{\emph{ICLR}}.
\newblock


\bibitem[\protect\citeauthoryear{Kumar, Kumar, Shah, and Faloutsos}{Kumar
  et~al\mbox{.}}{2018}]%
        {kumar2018did}
\bibfield{author}{\bibinfo{person}{Rohan Kumar}, \bibinfo{person}{Mohit Kumar},
  \bibinfo{person}{Neil Shah}, {and} \bibinfo{person}{Christos Faloutsos}.}
  \bibinfo{year}{2018}\natexlab{}.
\newblock \showarticletitle{Did We Get It Right? Predicting Query Performance
  in E-commerce Search}.
\newblock \bibinfo{journal}{\emph{arXiv preprint arXiv:1808.00239}}
  (\bibinfo{year}{2018}).
\newblock


\bibitem[\protect\citeauthoryear{Lamba and Shah}{Lamba and Shah}{2019}]%
        {lamba2019modeling}
\bibfield{author}{\bibinfo{person}{Hemank Lamba} {and} \bibinfo{person}{Neil
  Shah}.} \bibinfo{year}{2019}\natexlab{}.
\newblock \showarticletitle{Modeling Dwell Time Engagement on Visual
  Multimedia}. In \bibinfo{booktitle}{\emph{KDD}}. \bibinfo{pages}{1104--1113}.
\newblock


\bibitem[\protect\citeauthoryear{Lin, Althoff, and Leskovec}{Lin
  et~al\mbox{.}}{2018}]%
        {lin2018ll}
\bibfield{author}{\bibinfo{person}{Zhiyuan Lin}, \bibinfo{person}{Tim Althoff},
  {and} \bibinfo{person}{Jure Leskovec}.} \bibinfo{year}{2018}\natexlab{}.
\newblock \showarticletitle{I'll Be Back: On the Multiple Lives of Users of a
  Mobile Activity Tracking Application}. In \bibinfo{booktitle}{\emph{WWW}}.
  \bibinfo{pages}{1501--1511}.
\newblock


\bibitem[\protect\citeauthoryear{Liu, Shi, Pierce, and Ren}{Liu
  et~al\mbox{.}}{2019}]%
        {liu2019characterizing}
\bibfield{author}{\bibinfo{person}{Yozen Liu}, \bibinfo{person}{Xiaolin Shi},
  \bibinfo{person}{Lucas Pierce}, {and} \bibinfo{person}{Xiang Ren}.}
  \bibinfo{year}{2019}\natexlab{}.
\newblock \showarticletitle{Characterizing and Forecasting User Engagement with
  In-app Action Graph: A Case Study of Snapchat}. In
  \bibinfo{booktitle}{\emph{KDD}}.
\newblock


\bibitem[\protect\citeauthoryear{Lo, Frankowski, and Leskovec}{Lo
  et~al\mbox{.}}{2016}]%
        {lo2016understanding}
\bibfield{author}{\bibinfo{person}{Caroline Lo}, \bibinfo{person}{Dan
  Frankowski}, {and} \bibinfo{person}{Jure Leskovec}.}
  \bibinfo{year}{2016}\natexlab{}.
\newblock \showarticletitle{Understanding behaviors that lead to purchasing: A
  case study of pinterest}. In \bibinfo{booktitle}{\emph{KDD}}. ACM,
  \bibinfo{pages}{531--540}.
\newblock


\bibitem[\protect\citeauthoryear{Ma, Wang, Aggarwal, and Tang}{Ma
  et~al\mbox{.}}{2019a}]%
        {ma2019graph}
\bibfield{author}{\bibinfo{person}{Yao Ma}, \bibinfo{person}{Suhang Wang},
  \bibinfo{person}{Charu~C Aggarwal}, {and} \bibinfo{person}{Jiliang Tang}.}
  \bibinfo{year}{2019}\natexlab{a}.
\newblock \showarticletitle{Graph convolutional networks with eigenpooling}. In
  \bibinfo{booktitle}{\emph{KDD}}.
\newblock


\bibitem[\protect\citeauthoryear{Ma, Wang, Aggarwal, Yin, and Tang}{Ma
  et~al\mbox{.}}{2019b}]%
        {ma2019multi}
\bibfield{author}{\bibinfo{person}{Yao Ma}, \bibinfo{person}{Suhang Wang},
  \bibinfo{person}{Chara~C Aggarwal}, \bibinfo{person}{Dawei Yin}, {and}
  \bibinfo{person}{Jiliang Tang}.} \bibinfo{year}{2019}\natexlab{b}.
\newblock \showarticletitle{Multi-dimensional graph convolutional networks}. In
  \bibinfo{booktitle}{\emph{SDM}}.
\newblock


\bibitem[\protect\citeauthoryear{Papapetrou and Roussos}{Papapetrou and
  Roussos}{2014}]%
        {papapetrou2014social}
\bibfield{author}{\bibinfo{person}{Panagiotis Papapetrou} {and}
  \bibinfo{person}{George Roussos}.} \bibinfo{year}{2014}\natexlab{}.
\newblock \showarticletitle{Social context discovery from temporal app use
  patterns}. In \bibinfo{booktitle}{\emph{Ubicomp}}. \bibinfo{pages}{397--402}.
\newblock


\bibitem[\protect\citeauthoryear{Pope, Kolouri, Rostami, Martin, and
  Hoffmann}{Pope et~al\mbox{.}}{2019}]%
        {pope2019explainability}
\bibfield{author}{\bibinfo{person}{Phillip~E Pope}, \bibinfo{person}{Soheil
  Kolouri}, \bibinfo{person}{Mohammad Rostami}, \bibinfo{person}{Charles~E
  Martin}, {and} \bibinfo{person}{Heiko Hoffmann}.}
  \bibinfo{year}{2019}\natexlab{}.
\newblock \showarticletitle{Explainability Methods for Graph Convolutional
  Neural Networks}. In \bibinfo{booktitle}{\emph{CVPR}}.
  \bibinfo{pages}{10772--10781}.
\newblock


\bibitem[\protect\citeauthoryear{Qin, Song, Chen, Cheng, Jiang, and
  Cottrell}{Qin et~al\mbox{.}}{2017}]%
        {qin2017dual}
\bibfield{author}{\bibinfo{person}{Yao Qin}, \bibinfo{person}{Dongjin Song},
  \bibinfo{person}{Haifeng Chen}, \bibinfo{person}{Wei Cheng},
  \bibinfo{person}{Guofei Jiang}, {and} \bibinfo{person}{Garrison Cottrell}.}
  \bibinfo{year}{2017}\natexlab{}.
\newblock \showarticletitle{A dual-stage attention-based recurrent neural
  network for time series prediction}.
\newblock \bibinfo{journal}{\emph{arXiv preprint arXiv:1704.02971}}
  (\bibinfo{year}{2017}).
\newblock


\bibitem[\protect\citeauthoryear{Ribeiro, Singh, and Guestrin}{Ribeiro
  et~al\mbox{.}}{2016}]%
        {ribeiro2016should}
\bibfield{author}{\bibinfo{person}{Marco~Tulio Ribeiro},
  \bibinfo{person}{Sameer Singh}, {and} \bibinfo{person}{Carlos Guestrin}.}
  \bibinfo{year}{2016}\natexlab{}.
\newblock \showarticletitle{Why should i trust you?: Explaining the predictions
  of any classifier}. In \bibinfo{booktitle}{\emph{KDD}}. ACM,
  \bibinfo{pages}{1135--1144}.
\newblock


\bibitem[\protect\citeauthoryear{Schmitz, Aldrich, and Gouws}{Schmitz
  et~al\mbox{.}}{1999}]%
        {schmitz1999ann}
\bibfield{author}{\bibinfo{person}{Gregor~PJ Schmitz}, \bibinfo{person}{Chris
  Aldrich}, {and} \bibinfo{person}{Francois~S Gouws}.}
  \bibinfo{year}{1999}\natexlab{}.
\newblock \showarticletitle{ANN-DT: an algorithm for extraction of decision
  trees from artificial neural networks}.
\newblock \bibinfo{journal}{\emph{TNN}} \bibinfo{volume}{10},
  \bibinfo{number}{6} (\bibinfo{year}{1999}), \bibinfo{pages}{1392--1401}.
\newblock


\bibitem[\protect\citeauthoryear{Shah}{Shah}{2017}]%
        {shah2017flock}
\bibfield{author}{\bibinfo{person}{Neil Shah}.}
  \bibinfo{year}{2017}\natexlab{}.
\newblock \showarticletitle{Flock: Combating astroturfing on livestreaming
  platforms}. In \bibinfo{booktitle}{\emph{WWW}}. \bibinfo{pages}{1083--1091}.
\newblock


\bibitem[\protect\citeauthoryear{Shah, Lamba, Beutel, and Faloutsos}{Shah
  et~al\mbox{.}}{2017}]%
        {shah2017many}
\bibfield{author}{\bibinfo{person}{Neil Shah}, \bibinfo{person}{Hemank Lamba},
  \bibinfo{person}{Alex Beutel}, {and} \bibinfo{person}{Christos Faloutsos}.}
  \bibinfo{year}{2017}\natexlab{}.
\newblock \showarticletitle{The many faces of link fraud}. In
  \bibinfo{booktitle}{\emph{ICDM}}. IEEE, \bibinfo{pages}{1069--1074}.
\newblock


\bibitem[\protect\citeauthoryear{Shu, Cui, Wang, Lee, and Liu}{Shu
  et~al\mbox{.}}{2019}]%
        {shu2019defend}
\bibfield{author}{\bibinfo{person}{Kai Shu}, \bibinfo{person}{Limeng Cui},
  \bibinfo{person}{Suhang Wang}, \bibinfo{person}{Dongwon Lee}, {and}
  \bibinfo{person}{Huan Liu}.} \bibinfo{year}{2019}\natexlab{}.
\newblock \showarticletitle{defend: Explainable fake news detection}. In
  \bibinfo{booktitle}{\emph{KDD}}.
\newblock


\bibitem[\protect\citeauthoryear{Subramani and Rajagopalan}{Subramani and
  Rajagopalan}{2003}]%
        {subramani2003knowledge}
\bibfield{author}{\bibinfo{person}{Mani~R Subramani} {and}
  \bibinfo{person}{Balaji Rajagopalan}.} \bibinfo{year}{2003}\natexlab{}.
\newblock \showarticletitle{Knowledge-sharing and influence in online social
  networks via viral marketing}.
\newblock \bibinfo{journal}{\emph{CACM}} (\bibinfo{year}{2003}).
\newblock


\bibitem[\protect\citeauthoryear{Tang, Li, Sun, Yao, Mitra, and Wang}{Tang
  et~al\mbox{.}}{2020a}]%
        {tang2020transferring}
\bibfield{author}{\bibinfo{person}{Xianfeng Tang}, \bibinfo{person}{Yandong
  Li}, \bibinfo{person}{Yiwei Sun}, \bibinfo{person}{Huaxiu Yao},
  \bibinfo{person}{Prasenjit Mitra}, {and} \bibinfo{person}{Suhang Wang}.}
  \bibinfo{year}{2020}\natexlab{a}.
\newblock \showarticletitle{Transferring Robustness for Graph Neural Network
  Against Poisoning Attacks}. In \bibinfo{booktitle}{\emph{WSDM}}.
\newblock


\bibitem[\protect\citeauthoryear{Tang, Yao, Sun, Aggarwal, Mitra, and
  Wang}{Tang et~al\mbox{.}}{2020b}]%
        {tang2019joint}
\bibfield{author}{\bibinfo{person}{Xianfeng Tang}, \bibinfo{person}{Huaxiu
  Yao}, \bibinfo{person}{Yiwei Sun}, \bibinfo{person}{Charu Aggarwal},
  \bibinfo{person}{Prasenjit Mitra}, {and} \bibinfo{person}{Suhang Wang}.}
  \bibinfo{year}{2020}\natexlab{b}.
\newblock \showarticletitle{Joint Modeling of Local and Global Temporal
  Dynamics for Multivariate Time Series Forecasting with Missing Values}.
\newblock  (\bibinfo{year}{2020}).
\newblock


\bibitem[\protect\citeauthoryear{Trouleau, Ashkan, Ding, and Eriksson}{Trouleau
  et~al\mbox{.}}{2016}]%
        {trouleau2016just}
\bibfield{author}{\bibinfo{person}{William Trouleau}, \bibinfo{person}{Azin
  Ashkan}, \bibinfo{person}{Weicong Ding}, {and} \bibinfo{person}{Brian
  Eriksson}.} \bibinfo{year}{2016}\natexlab{}.
\newblock \showarticletitle{Just one more: Modeling binge watching behavior}.
  In \bibinfo{booktitle}{\emph{KDD}}. ACM, \bibinfo{pages}{1215--1224}.
\newblock


\bibitem[\protect\citeauthoryear{Vaswani, Shazeer, Parmar, Uszkoreit, Jones,
  Gomez, Kaiser, and Polosukhin}{Vaswani et~al\mbox{.}}{2017}]%
        {vaswani2017attention}
\bibfield{author}{\bibinfo{person}{Ashish Vaswani}, \bibinfo{person}{Noam
  Shazeer}, \bibinfo{person}{Niki Parmar}, \bibinfo{person}{Jakob Uszkoreit},
  \bibinfo{person}{Llion Jones}, \bibinfo{person}{Aidan~N Gomez},
  \bibinfo{person}{Lukasz Kaiser}, {and} \bibinfo{person}{Illia
  Polosukhin}.} \bibinfo{year}{2017}\natexlab{}.
\newblock \showarticletitle{Attention is all you need}. In
  \bibinfo{booktitle}{\emph{NeurIPS}}. \bibinfo{pages}{5998--6008}.
\newblock


\bibitem[\protect\citeauthoryear{Xu, Biswal, Deshpande, Maher, and Sun}{Xu
  et~al\mbox{.}}{2018}]%
        {xu2018raim}
\bibfield{author}{\bibinfo{person}{Yanbo Xu}, \bibinfo{person}{Siddharth
  Biswal}, \bibinfo{person}{Shriprasad~R Deshpande}, \bibinfo{person}{Kevin~O
  Maher}, {and} \bibinfo{person}{Jimeng Sun}.} \bibinfo{year}{2018}\natexlab{}.
\newblock \showarticletitle{Raim: Recurrent attentive and intensive model of
  multimodal patient monitoring data}. In \bibinfo{booktitle}{\emph{KDD}}. ACM,
  \bibinfo{pages}{2565--2573}.
\newblock


\bibitem[\protect\citeauthoryear{Yang, Shi, Jie, and Han}{Yang
  et~al\mbox{.}}{2018}]%
        {yang2018know}
\bibfield{author}{\bibinfo{person}{Carl Yang}, \bibinfo{person}{Xiaolin Shi},
  \bibinfo{person}{Luo Jie}, {and} \bibinfo{person}{Jiawei Han}.}
  \bibinfo{year}{2018}\natexlab{}.
\newblock \showarticletitle{I Know You'll Be Back: Interpretable New User
  Clustering and Churn Prediction on a Mobile Social Application}. In
  \bibinfo{booktitle}{\emph{KDD}}. ACM, \bibinfo{pages}{914--922}.
\newblock


\bibitem[\protect\citeauthoryear{Yang, Wei, Ackerman, and Adamic}{Yang
  et~al\mbox{.}}{2010}]%
        {yang2010activity}
\bibfield{author}{\bibinfo{person}{Jiang Yang}, \bibinfo{person}{Xiao Wei},
  \bibinfo{person}{Mark~S Ackerman}, {and} \bibinfo{person}{Lada~A Adamic}.}
  \bibinfo{year}{2010}\natexlab{}.
\newblock \showarticletitle{Activity lifespan: An analysis of user survival
  patterns in online knowledge sharing communities}. In
  \bibinfo{booktitle}{\emph{ICWSM}}.
\newblock


\bibitem[\protect\citeauthoryear{Yao, Tang, Wei, Zheng, and Li}{Yao
  et~al\mbox{.}}{2019}]%
        {yao2019revisiting}
\bibfield{author}{\bibinfo{person}{Huaxiu Yao}, \bibinfo{person}{Xianfeng
  Tang}, \bibinfo{person}{Hua Wei}, \bibinfo{person}{Guanjie Zheng}, {and}
  \bibinfo{person}{Zhenhui Li}.} \bibinfo{year}{2019}\natexlab{}.
\newblock \showarticletitle{Revisiting spatial-temporal similarity: A deep
  learning framework for traffic prediction}. In
  \bibinfo{booktitle}{\emph{AAAI}}.
\newblock


\bibitem[\protect\citeauthoryear{Ying, Bourgeois, You, Zitnik, and
  Leskovec}{Ying et~al\mbox{.}}{2019}]%
        {ying2019gnn}
\bibfield{author}{\bibinfo{person}{Rex Ying}, \bibinfo{person}{Dylan
  Bourgeois}, \bibinfo{person}{Jiaxuan You}, \bibinfo{person}{Marinka Zitnik},
  {and} \bibinfo{person}{Jure Leskovec}.} \bibinfo{year}{2019}\natexlab{}.
\newblock \showarticletitle{GNN Explainer: A Tool for Post-hoc Explanation of
  Graph Neural Networks}. In \bibinfo{booktitle}{\emph{NeurIPS}}.
\newblock


\bibitem[\protect\citeauthoryear{Ying, He, Chen, Eksombatchai, Hamilton, and
  Leskovec}{Ying et~al\mbox{.}}{2018}]%
        {ying2018graph}
\bibfield{author}{\bibinfo{person}{Rex Ying}, \bibinfo{person}{Ruining He},
  \bibinfo{person}{Kaifeng Chen}, \bibinfo{person}{Pong Eksombatchai},
  \bibinfo{person}{William~L Hamilton}, {and} \bibinfo{person}{Jure Leskovec}.}
  \bibinfo{year}{2018}\natexlab{}.
\newblock \showarticletitle{Graph convolutional neural networks for web-scale
  recommender systems}. In \bibinfo{booktitle}{\emph{KDD}}.
  \bibinfo{pages}{974--983}.
\newblock


\bibitem[\protect\citeauthoryear{Zilke, Menc{\'\i}a, and Janssen}{Zilke
  et~al\mbox{.}}{2016}]%
        {zilke2016deepred}
\bibfield{author}{\bibinfo{person}{Jan~Ruben Zilke},
  \bibinfo{person}{Eneldo~Loza Menc{\'\i}a}, {and} \bibinfo{person}{Frederik
  Janssen}.} \bibinfo{year}{2016}\natexlab{}.
\newblock \showarticletitle{DeepRED--Rule extraction from deep neural
  networks}. In \bibinfo{booktitle}{\emph{ICDS}}. Springer,
  \bibinfo{pages}{457--473}.
\newblock


\end{thebibliography}
\clearpage
\appendix
\section{Complexity of \ours}
\subsection{Theoretical Analysis} \label{apd:complexity} 
In this section, we analyze the complexity of \ours. In particular, we focus on the complexity reduction from the tensor-based designs of GCN and LSTM over the standard ones. 
Without loss of generality, we use $d_{\text{in}}$ and $d_{\text{out}}$ to denote the dimensions of input and output of a neural network layer (e.g., GCN, LSTM, etc.).
We use the number of learnable parameters (neurons in the network) to measure the network complexity as follows:
\begin{theorem} \label{theo:space}
By replacing the standard GCN and LSTM layers with corresponding tensor-based versions, the network complexity is reduced by $(1-\frac{1}{K})d_{\text{in}} \cdot d_{\text{out}}$ and $4(1-\frac{1}{K})(d_{\text{in}} + d_{\text{out}})d_{\text{out}}$ number of trainable parameters, respectively.
\end{theorem}

\begin{proof}
The number of trainable parameters for the GCN layer is $d_{\text{in}} \cdot d_{\text{out}}$ (see Eqn. \ref{eqn:gcn}), and that for the tensor-based GCN layer is $K\cdot(\frac{d_{\text{in}}}{K} \cdot \frac{d_{\text{out}}}{K}) = \frac{d_{\text{in}} \cdot d_{\text{out}}}{K}$ (see Eqn. \ref{eqn:tgcn}, assume they are equally divided into each category of user action features). Therefore, tensor-based GCN reduces network complexity by $(1-\frac{1}{K})d_{\text{in}} \cdot d_{\text{out}}$ number of parameters.
Similarly, the standard LSTM layer has $4(d_{\text{in}} \cdot d_{\text{out}} + d_{\text{out}}^2 + d_{\text{out}})$ trainable parameters (corresponding to the input transition, hidden state transition, and the bias); while the tensor-based LSTM layer only maintains $4(\frac{d_{\text{in}} \cdot d_{\text{out}}}{K} + \frac{d_{\text{out}}^2}{K} + d_{\text{out}})$ number of parameters (for $\mathcal{U}_*$, $\mathcal{U}_*^\h$ and $\b_*$ in Eqn. \ref{eqn:lstm}). As a result, the total number of parameters is reduced by $4(1-\frac{1}{K})(d_{\text{in}} + d_{\text{out}})d_{\text{out}}$ when adopting the tensor-based LSTM over the standard one.
\end{proof}

The computational complexity comes from multiplications. 
The reduction of computational complexity is analyzed through Theorem \ref{theo:time}:
\begin{theorem} \label{theo:time}
The tensor-based GCN and the tensor-based LSTM reduce the computational complexity by $\mathcal{O}(d_{\text{in}} \cdot d_{\text{out}})$ and $\mathcal{O}((d_{\text{in}} + d_{\text{out}})d_{\text{out}})$, respectively.
\end{theorem}
\begin{proof}
Let $N$ denote the number of nodes in the ego-network.
Using Eqn. \ref{eqn:gcn} and \ref{eqn:tgcn}, the computational complexity of a GCN layer and a tensor-based GCN layer are $N^2 \cdot d_{\text{in}} + N \cdot d_{\text{in}} \cdot d_{\text{out}}$ and $N^2 \cdot \frac{d_{\text{in}}}{K} \cdot K + N  \cdot \frac{d_{\text{in}}}{K}  \cdot \frac{d_{\text{out}}}{K} \cdot K = N^2 \cdot d_{\text{in}} + N  \cdot \frac{d_{\text{in}} \cdot d_{\text{out}}}{K}$, respectively. The reduction is then $N(1-\frac{1}{K})d_{\text{in}} \cdot d_{\text{out}} = \mathcal{O}(d_{\text{in}} \cdot d_{\text{out}})$.
For an LSTM layer (Eqn. \ref{eqn:lstm}), it takes $4(d_{\text{in}} \cdot d_{\text{out}} + d_{\text{out}}^2) + 3d_{\text{out}}$ multiplications to update its hidden and gate, while the tensor-based LSTM layer takes only $4(\frac{d_{\text{in}}}{K}  \cdot \frac{d_{\text{out}}}{K} \cdot K + \frac{d_{\text{in}}^2}{K^2} \cdot K) + 3d_{\text{out}} = 4(\frac{d_{\text{in}} \cdot d_{\text{out}}}{K} + \frac{d_{\text{out}}^2}{K}) + 3d_{\text{out}}$ multiplications. Thus, the reduction of computational complexity by the tensor-based LSTM is $\mathcal{O}((d_{\text{in}} + d_{\text{out}})d_{\text{out}})$.
\end{proof}

Note that for \ours, it adopts multiple friendship modules with the tensor-based LSTM. Therefore, \ours is significantly benefited from the tensor-based design, reducing both network size and computational complexity sharply. However, the overall improvement over complexity does not exactly aligned with these tensor-based designs, due to costs from extra components in \ours such as the computation of attention scores. Therefore, we also analyze the real-world running time of \ours quantitatively in the following experiment section.

\subsection{Experimental Results}
\label{apd:runtime_experiments}
We study the runtime complexity of \ours. We compare the runtime of $\text{\ours}_{ts}$ and \ours, to demonstrate the improvement by using tensor-based designs for \ours over a non tensor-based model $\text{\ours}_{ts}$
Both training and testing (inference) run times are reported in \ref{tab:runtime}.
We can see that training \ours takes significantly less time then $\text{\ours}_{ts}$ by an average of 20\%.
In addition, inference speed of \ours is also faster. Therefore, it is beneficial to adopt tensor-based designs when constructing the framework. Note that our implementation uses PyTorch Geometric\footnote{https://github.com/rusty1s/pytorch\_geometric} as the underlying message passing framework. 
\begin{table}[ht]
    \centering
    \caption{Comparisons of Runtime ($min$). \ours reduced 20\% of runtime on average comparing with non-tensor-based $\text{\ours}_{ts}$.
    }
    \vskip -1em
\begin{tabular}{cccccc}
\toprule
\multicolumn{2}{c}{\multirow{2}{*}{}}  & \multicolumn{2}{c}{ \regionone} & \multicolumn{2}{c}{\regiontwo} \\ 
\cmidrule(lr){3-4}
\cmidrule(lr){5-6}
\multicolumn{2}{c}{}  & \textbf{Train} & \textbf{Test} & \textbf{Train} & \textbf{Test} \\ \midrule
\multirow{2}{*}{\rotatebox[origin=c]{0}{\taskone}} & $\textbf{\ours}_{ts}$    &  216.96   &  148.20 &  132.25  & 90.42 \\
                        & \textbf{\ours}                  &  \textbf{181.35}  & \textbf{119.40}  & \textbf{117.70}  & \textbf{73.43} \\ \midrule
\multirow{2}{*}{\rotatebox[origin=c]{0}{\tasktwo}}  & $\textbf{\ours}_{ts}$     &207.23   & 151.48 &  137.50 & 89.61 \\
                        & \textbf{\ours}                  &   \textbf{172.00}  & \textbf{115.10} & \textbf{110.92}  & \textbf{69.05}  \\  \bottomrule 
\end{tabular}
    \label{tab:runtime}
\end{table}

\section{Implementation Details} \label{imp_details}

\subsection{Experimental Environment}
Our experiments are conducted on a single machine on Google Cloud Platform\footnote{https://cloud.google.com}, with a 16-core CPU, 60GB memory and 2 Nvidia P100 GPUs.

\subsection{Data Preprocessing}
\begin{table}[b]
\setlength\tabcolsep{4pt}
\small
    \centering
    \caption{Statistics of Datasets.}
    \vskip -1em
    \begin{tabular}{cccccc} 
    \toprule

         &  \regionone & \regiontwo \\ \midrule
     \textbf{Time period}    & \multicolumn{2}{c}{09/16/2019 - 10/27/2019} \\
     \textbf{Avg. \# users} & 153006 & 108452 \\
     \textbf{Avg. node degree} & 51.58 & 36.95 \\
     \textbf{\# node features} & \multicolumn{2}{c}{13} \\
     \textbf{\# edge features} & \multicolumn{2}{c}{3} \\ \bottomrule
    \end{tabular}
    \label{tab:data}
\end{table}

\begin{table*}[!h]
    \centering
    \caption{Selected features for user actions on \snap.}
    \vspace{-1em}
    \begin{tabular}{ccl}
    \toprule
In-App function                            & Feature name    & \multicolumn{1}{c}{Description}      \\ \midrule
\multirow{4}{*}{Snap}                      & \feature{SnapSend}       & \# of snaps sent to friends.                            \\ 
                                           & \feature{SnapView}        & \# of snaps viewed from friends.                        \\
                                           & \feature{SnapCreate}      & \# of snaps created by the user.                        \\
                                           & \feature{SnapSave}        & \# of snaps saved to the memory/smartphone.             \\  \midrule
\multirow{2}{*}{Chat}                      & \feature{ChatSend}        & \# of text messages sent to friends.                    \\
                                           & \feature{ChatView}        & \# of received text messages.                           \\  \midrule
\multicolumn{1}{c}{\multirow{3}{*}{Story}} & \feature{StoryPost}       & \# of videos posted to the user's page                  \\
                       & \feature{StoryView}       & \# of watched story videos posted by others.            \\
                     & \feature{StoryViewTime}  & Total time spent for watching stories.                  \\  \midrule
\multirow{3}{*}{Discover}                  & \feature{FriendDiscoverView}  & \# of watched videos posted by friends on Discover page \\
                                           & \feature{PublisherDiscoverView} & \# of watched videos posted by publisher on Discover page  \\
                                           & \feature{DiscoverViewTime}  & Total time spent for watching videos on Discover page.  \\ \midrule
Misc.                                          & \feature{SessionTime}       & Total time spent on Snapchat.  \\ \bottomrule  
\end{tabular}
    \label{tab:action}
\end{table*}

We select two geographic regions, one from North America and the other from Europe, to compile two datasets.
We set the time period from 09/16/2019 to 10/27/2019, with a one-day time interval length. There are totally 42 days in the time period (6 weeks).
For each dataset, we first query all users whose locations are within the corresponding region. 
Users who spend less than one minute (session time) on a daily average are treated as extremely inactive and filtered.
We then obtain the friendship of these users as our social network and historical user action records in each day.
Detailed features and descriptions for user actions are reported in Table \ref{tab:action}.
Besides, we also acquire user-user commutation as features for user interaction, including chat, snap, and story. These features are constructed from the aggregation of each type of interaction.
Table \ref{tab:data} details both datasets.

\subsection{Model Implementations}
We implement all compared baseline methods in Python 3.7. 
Linear Regression is adopted from scikit-learn\footnote{https://scikit-learn.org}.
We use XGBoost\cite{chen2016xgboost} from the official package\footnote{https://xgboost.readthedocs.io/} with its recommended setting and parameters.
We implement the GCN model with PyTorch Geometric. We set up a two layer GCN, with the hidden size of 128, using ELU as the activation function.
Similarly, we build the LSTM model as a two-layer LSTM using PyTorch\footnote{https://pytorch.org/}. The hidden size is 128. We set the dropout rate to 0.5 for the second layer. ELU is used as the activation.
We following the original settings for TGLSTM as introduced in the paper \cite{liu2019characterizing}.
We implement \ours with PyTorch and PyTorch Geometric. 
Friendship modules contain two-layer \tgcn. The dimension of output embedding for all feature categories is set to 32. 
The design of \tlstm is inspired by IMV-LSTM\footnote{https://github.com/KurochkinAlexey/IMV\_LSTM}. We use two layers of \tlstm for \ours.
Our code is available on  \textbf{Github}\footnote{https://github.com/tangxianfeng/FATE}.

For LR and Xgboost, we train until convergence. 
For neural network models, we set the batch size to 256 and the max number of epoch to 10.
All models are optimized by Adam algorithm \cite{kingma2014adam}, with a learning rate of 0.001.
They are trained until reaching the max epoch or early-stopped on the validation set. The validation set contains 10\% samples randomly selected from the training set.
All methods are trained and tested 10 times to get averaged results.

\subsection{Evaluation Metrics} \label{evaluation_metrics}
Three common metrics Root Mean Square Error (RMSE), Mean Absolute Percentage Error (MAPE) and  Mean Absolute Error (MAE) are used to evaluate the performance of all methods.
The detailed definitions of these metrics are stated as below:

\begin{align}
    \nonumber \text{RMSE} &=\sqrt{\frac{1} {|\S|}\sum_{u\in\S} (e^u-\hat{e}^u)^2}, \\ 
    \nonumber \text{MAPE} &=\frac{1}{|\S|} \sum_{u\in\S} \frac{|e^u-\hat{e}^u|}{\hat{e}^u}, \\
    \text{MAE} & = \frac{1}{|\S|} \sum_{u\in\S} {|e^u-\hat{e}^u|},
\end{align}
where $\hat{e}^u$ denotes the ground truth of predicted user engagement score $e^u$.

While RMSE and MAE receive higher penalties from larger values, MAPE focuses on the prediction error of samples with smaller engagement scores. 
Therefore, combining these metrics leads to more comprehensive conclusions.

\end{document}